\documentclass[12pt]{article} 

\usepackage[ngerman,english]{babel}
\useshorthands{"}
\addto\extrasenglish{\languageshorthands{ngerman}}
\usepackage[utf8]{inputenc}
\usepackage[T1]{fontenc}
\usepackage{amsmath}
\usepackage{graphicx}
\usepackage{amssymb}
\usepackage{hyperref}
\usepackage{amsmath}
\usepackage{mathrsfs}
\usepackage{color}
\usepackage{subfig}
\usepackage{amsthm}

\newtheorem{definition}{Definition}[section]

\newtheorem{lemma}{Lemma}[subsection]

\theoremstyle{definition}
\newtheorem{example}{Example}[section]

\renewcommand{\theta}{\vartheta}
\newcommand{\x}{\mathbf x}
\newcommand{\y}{\mathbf{b}}

\title{{\Large\bf Framework for the Quantum Mechanical Sum of Possibilities and Meaning for Field Theory and Gravity}}
\author{{\bf Artem Averin$^{\textrm{a}}$\footnote{artem.averin@campus.lmu.de}}}

\begin{document}

\maketitle

\centerline{\it $^{\textrm{a}}$ Arnold--Sommerfeld--Center for Theoretical Physics,}
\centerline{\it Ludwig--Maximilians--Universit\"at, 80333 M\"unchen, Germany}

\vskip1cm
\begin{abstract}
{{
In quantum mechanics, the measureable quantities of a given theory are predicted by performing a weighted sum over possibilities. We show how to arrange the possibilities into bundles such that the associated subsums can be viewed as well-defined theories on their own right. These bundles are submani\emph{folds} of \emph{possi}bilities which we call possifolds. We collect and prove some basic facts about possifolds. Especially, we show that possifolds are ensembles of what in a certain broadly defined sense that we explain can be regarded as soliton excitations (soliton-possifold correspondence). We provide an outlook on some applications. Among other things, we illustrate the use of the developed framework for the example of the Lieb-Liniger model. It describes non-relativistic bosons with an attractive interaction. We derive a dual theory describing the lowest-lying energy excitation modes. While the standard Bogoliubov-approximation breaks down at the critical point, our derived summation prescription stays regular. In the Bogoliubov-limit we observe the summation to possess an enhanced symmetry at this point while the summation cannot be ignored there. We finally provide a glimpse on the restrictions black hole physics implies in this context for the gravitational path integral.       
}}
\end{abstract}


\newpage

\setcounter{tocdepth}{2}
\tableofcontents
\break

\section{Introduction: General Idea}
\label{Kapitel 1}

It is the sum of possibilities. This, perhaps, summarizes most simply what contemporary fundamental physics is about. Let us explain this in an example. Suppose, we let a coin fall down and we are interested in observing it by measuring an observable $O$ a time $T$ after the coin's fall. How do we make a prediction for $O$?

In classical mechanics, the coin's trajectory is determined by Newton's law and allows for the computation of $O.$ This is schematically depicted in Fig. \ref{Fig. 1} (a).  

\begin{figure}[h!]
  \centering
  \subfloat[][]{\includegraphics[width=0.25\linewidth]{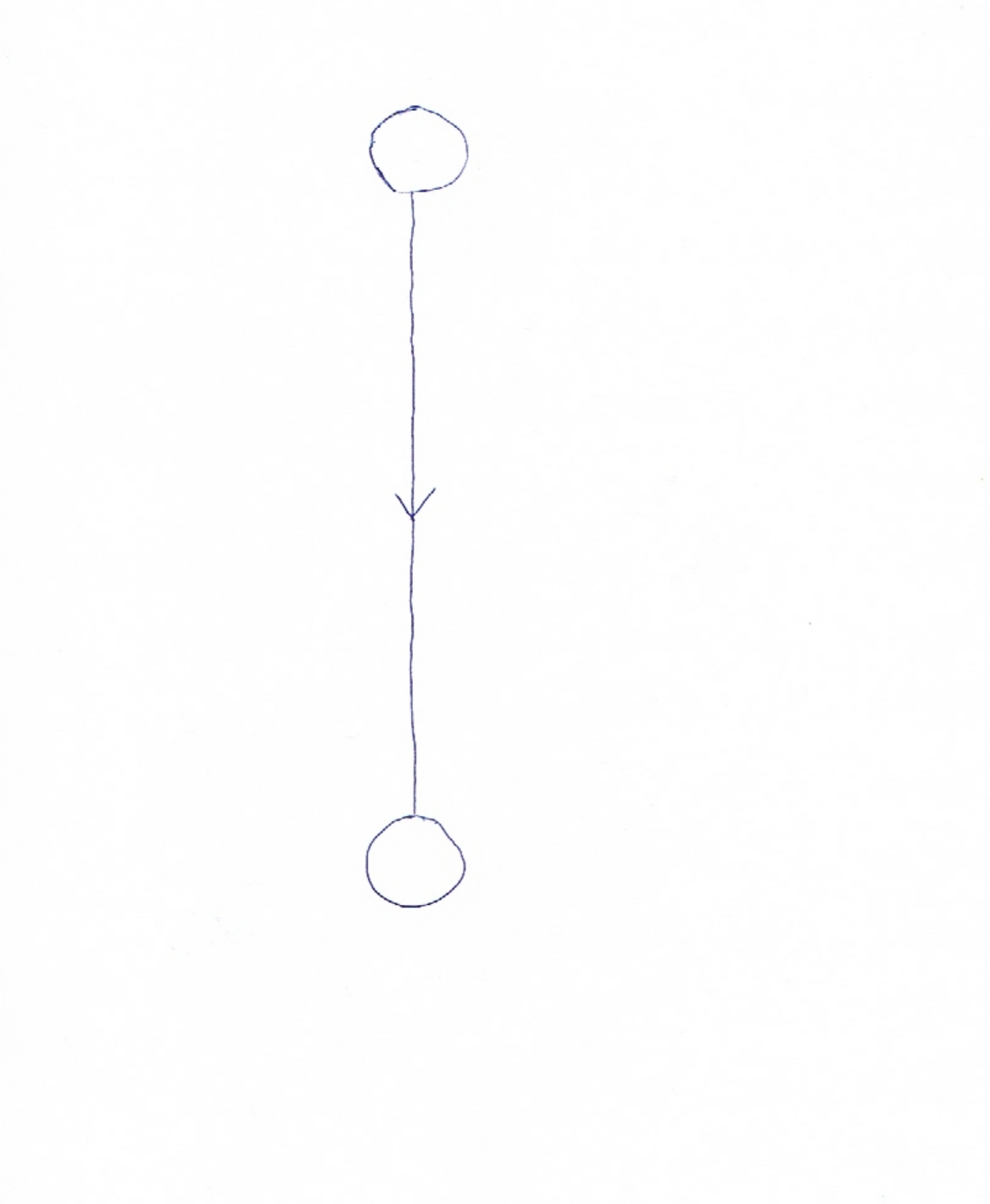}}%
  \qquad
  \subfloat[][]{\includegraphics[width=0.25\linewidth]{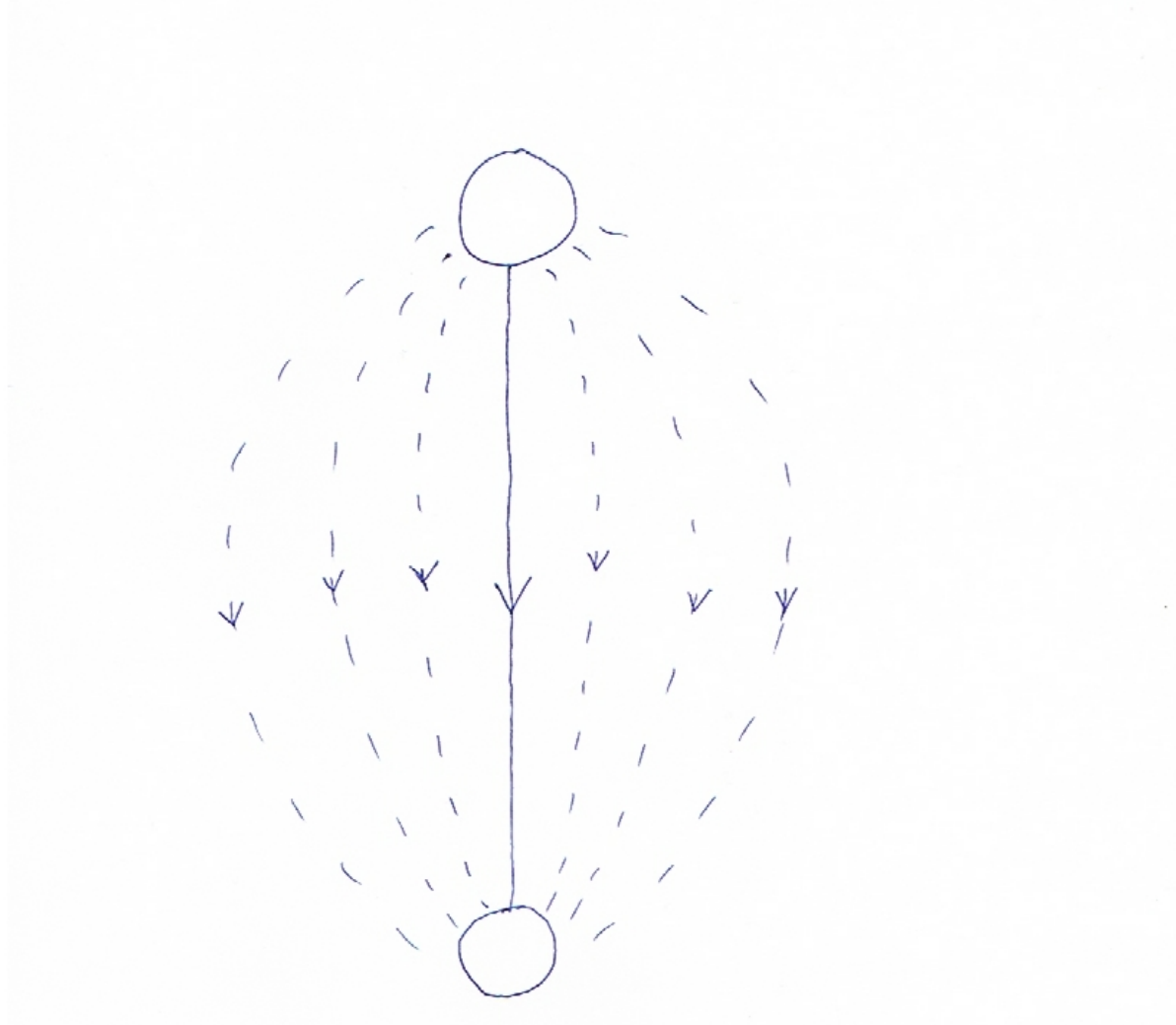}}%
	\qquad
  \subfloat[][]{\includegraphics[width=0.25\linewidth]{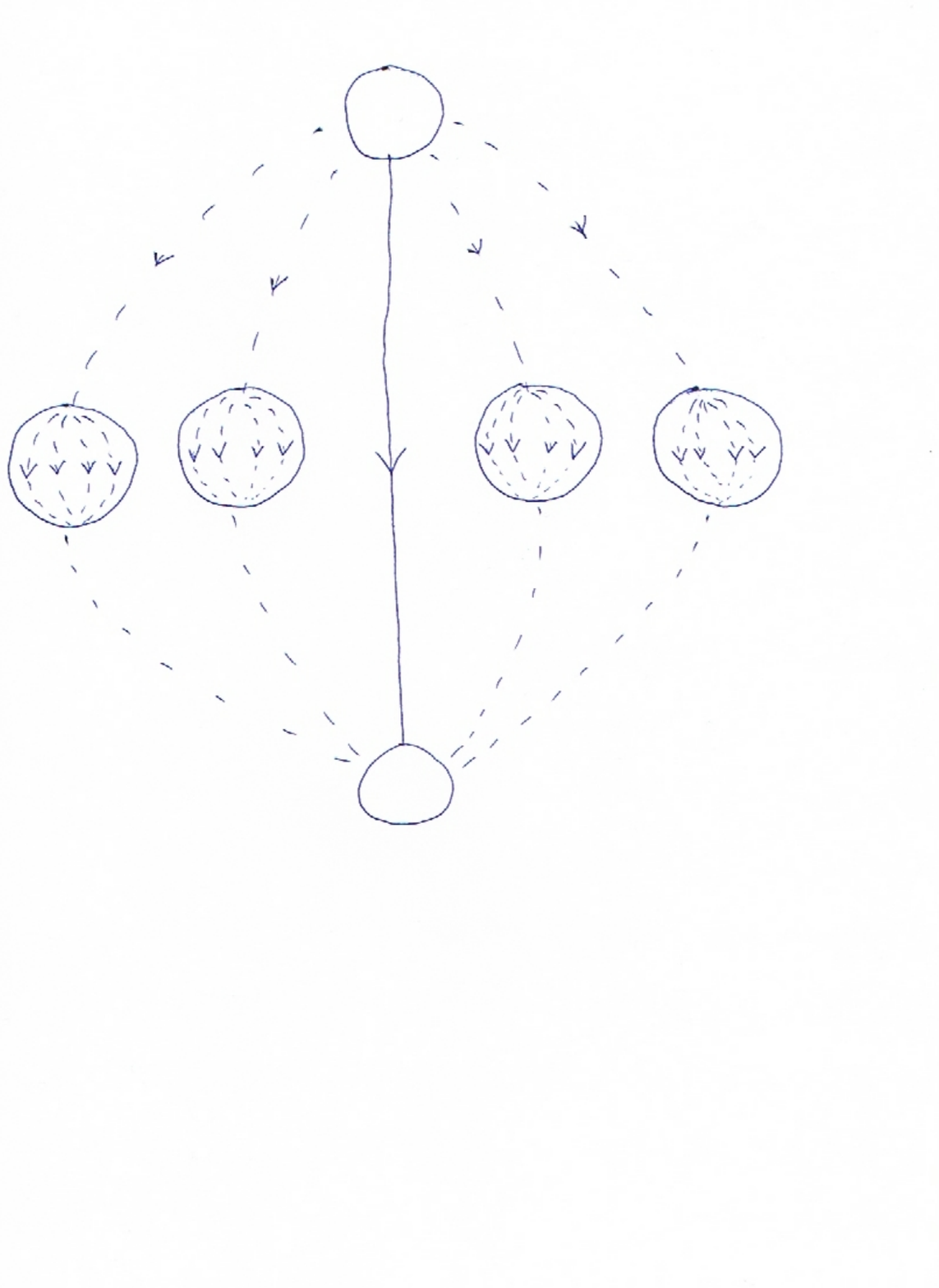}}%
  \caption{(a) Classically, equations of motion determine a unique trajectory. (b) Quantum mechanically, a sum over all possibilities has to be performed. Apart from the classical trajectory (solid line), every other path (dashed lines) contributes. (c) Various possibilities in the sum are packed together into bundles (possifolds). They are shown as the paths inside the circles.}%
	\label{Fig. 1}
\end{figure} 

In the present example, the predictions obtained this way are quite accurate. However, for over a century it is now known that for a proper description of Nature classical mechanics has to be replaced by quantum mechanics. In quantum mechanics, classical equations of motion are replaced by the Schrödinger equation. Instead of considering a single trajectory, the latter instructs us to form a weighted sum over all possible trajectories the coin could take. 

Schematically, this reads as

\begin{equation}
O = \sum_{\text{possibilities}}{(\text{weight factor})}
\label{1}
\end{equation}

and the situation is shown in Fig. \ref{Fig. 1} (b). 

Of course, in the given example, most of the trajectories in Fig. \ref{Fig. 1} (b) interfere destructively and the dominant contribution in \eqref{1} comes from the classical trajectory. However, for our argument this will be not important. What matters, is that the quantum corrections are there as a matter of principle. Although we could have chosen every other process, we will stay and quite often resort to this particular simple example for our explanations. In fact, what does it even mean to consider a concrete process?

To specify a concrete theory under consideration, we have to define what the ingredients of \eqref{1} are. Hence, what we will call a \emph{theory,} is a choice $(\text{possibilities}, \text{ weight factors}).$ We will explain the concrete mathematical meaning (and why it follows from the Schrödinger equation) in the next chapter. We use here only a schematic notation in order to bring across main ideas and the plan of our program. The weight factors associate a number to each possibility over which we have to be able to sum. A particular choice of such a theory is then via \eqref{1} related to its predicted observables. If two theories $(\text{possibilities}, \text{ weight factors})$ and $(\text{possibilities}', \text{ weight factors}')$ can be mapped onto each other such that all predicted observables are equal, we will speak of a \emph{duality}.

Obviously, the sum in \eqref{1} is commutative. That is, the order of summation over trajectories in Fig. \ref{Fig. 1} (b) can be chosen at will. This allows us to pack various possibilities into bundles in the summation \eqref{1}. This is visualized in Fig. \ref{Fig. 1} (c). 

In these kind of bundles, we are summing over a subset of the possibilities in the theory with appropriate weight factors. The idea is to consider such bundles which can comprise a theory (in the sense discussed above) on its own right. To specify such a bundle, we have to choose an appropriate subset of possibilities of the original theory. We have to insist that we are able to sum over this subset, i.e. we must be able to lable or coordinatize it. Hence, it must be a proper (sub)mani\emph{fold} of the \emph{possi}bilities. In analogy with word creations as orbifolds, orientifolds, Calabi-Yau $n$-folds and others, we will denote such objects as \emph{possifolds}. They will form our main topic here. 

To get a more rigorous understanding of the schematic ideas presented here, we remember that the basic ingredient in the heuristic definition of a possifold was the commutativity of the sum \eqref{1}. The commutativity of a discrete sum corresponds in the continuum limit to the covariance of an integral (i.e. its invariance under coordinate transformations). Our task is hence clear. Despite of giving a clear meaning to \eqref{1} as already mentioned, we have to rewrite the (functional) integral in \eqref{1} in a manifestly covariant way. This is precisely what we will do in the next chapter. Motivated by the discussion in this chapter, this will then allow us to give a rigorous definition and discussion of the notion of a possifold.

Before going into this task, let us pause for a moment and explore why possifolds should be discussed. What is their physical meaning?

As described, the idea behind possifolds is to reduce a given theory to a smaller and simpler one. Hereby, an allowed reduction corresponds to a possifold. For a given theory its possifolds thus provide ``lines of cuts'' where reductions can take place. They thus provide an orientation tool to navigate within a given theory. Such a navigation should be invariant under dualities (in the sense dualities were discussed above). 

Dualities are prominently present in string theory. There, the tool used to find and navigate through the various dualities among the different string theories are $D$-branes. The latter are solitonic configurations with certain characteristics (like their tension, excitation spectrum and its supersymmetry). The standard argument for uncovering a duality among two string theories proceeds by comparing the characteristics of their $D$-branes respectively. From our discussion so far, we might therefore expect in general a relation between the possifolds of a given theory and its solitonic configurations. Indeed, we will find the \emph{soliton-possifold correspondence} stating that a possifold is in one-to-one correspondence with an ensemble of excitations of a soliton in a given theory. 

Conversely, this means that a proper description of a solitonic excitation spectrum is restricted to form a possifold. Only in that case, the soliton excitations can be described by what can be called a dual theory. Hence, we learn that the possifold notion provides a systematic tool for the study of solitons and their excitations. In fact, the possifold notion was encountered in this context in \cite{Averin:2018owq,Averin:2019zsi} in the analysis of black holes. However, here, we will not go into this analysis. Rather, our purpose here is to develop the possifold notion from first principles and justify basic facts. 

So far, we have presented the possifold notion merely as a dissection tool for a given theory $(\text{possibilities}, \text{ weight factors})$ in order to explore its possibilities in a systematic way. However, as explained, quantum mechanics is the sum of possibilities. Can possifolds also teach us something about this sum? 

In physics, it happens quite often that a certain quantity of interest is given explicitly by a perturbation series. Under circumstances it may, however, happen that the various orders of this series become comparable causing a breakdown of the perturbative expansion. Quite often, this breakdown can be traced back to an important physical effect. Taking into account this effect properly in the perturbative expansion allows in some cases a resummation into a new more useful perturbative expansion. For instance, external soft gluon legs in hadronic processes are resummed by the introduction of the parton distribution functions. Loop diagrams in effective vertices leading to large logarithms are resummed by taking into account the renormalization group flow of couplings. We see, the maxim ``to make perturbation theory work again'' is well-known in physics. In chapter \ref{Kapitel 4.3}, we would like to ask whether possifolds can be added as a further tool to follow this maxim. Since possifolds work as anchor points in the sum of possibilities in a given theory, the sensitivity of this sum on a particular set of possibilities is expected to be controlled by the properties of near lying possifolds. This may explain the presence or absence of particular quantum mechanical sensitivities which are not visible in usual perturbative expansions. We will give a glimpse of this line of thought in chapter \ref{Kapitel 4.3}.

To summarize, we have provided an intuitive motivation for what a possifold should be and why this notion should be studied. As is quite common in quantum field theory, while the basic concepts are very descriptive, their rigorous formulation is rather technical. The purpose of this introduction was to provide a vivid description of the program and technicalities that are going to come.

The paper is organized as follows. In the next chapter \ref{Kapitel 2}, we will give as promised a rigorous meaning and manifestly covariant formulation of the sum of possibilities \eqref{1}. In chapter \ref{Kapitel 3} this will motivate the possifold notion. Here, we also summarize and prove some basic facts concerning this notion. In chapter \ref{Kapitel 4} we will give some applications to justify the utility of the developed framework. In particular, we will give some examples of the use of the involved formalism for the case of the Lieb-Liniger model in chapter \ref{Kapitel 4.1}. In chapter \ref{Kapitel 4.2} we establish the mentioned soliton-possifold correspondence. Its meaning for the problem of quantum mechanical sensitivity is discussed in chapter \ref{Kapitel 4.3}.

\section{The Sum of Possibilities}
\label{Kapitel 2}

\subsection{Manifestly Covariant Prescription}
\label{Kapitel 2.1}

Consider a system with generalized coordinates $q=(q_1,\ldots,q_N)$ and momenta $p=(p_1,\ldots,p_N).$ We denote the phase space by $\Gamma = \{ (q,p) \}.$ The classical observables are (smooth) maps $O: \Gamma \longrightarrow \mathbb{R}$ and are associated upon quantization to quantum mechanical operators $O: \mathcal{H} \longrightarrow \mathcal{H}$ on the Hilbert-space $\mathcal{H}.$ For a given Hamiltonian $H=H(q,p)$ the basic measurable quantities in the quantum theory are given by the correlation functions of quantum mechanical operators. For observables $O_i : \Gamma \longrightarrow \mathbb{R}$ and times $t_i$ for $i=1,\ldots,n$ they can be predicted by the formula

\begin{equation} \label{2}
\begin{split}
&\langle \Omega | TO_1^H (t_1) \cdots O_n^H(t_n) |\Omega \rangle \\
&=\text{\begin{tiny}$\lim_{T \to \infty(1-i\varepsilon)} \frac{\int \mathcal{D}q(t) \mathcal{D}p(t) O_1(q(t_1),p(t_1)) \cdots O_n(q(t_n),p(t_n)) 
e^{\frac{i}{\hbar} \int_{-T}^T dt \left(\sum_i p^i(t) \dot{q}^i(t) -H(q(t),p(t)) \right)}}
{\int \mathcal{D}q(t) \mathcal{D}p(t)  
e^{\frac{i}{\hbar} \int_{-T}^T dt \left(\sum_i p^i(t) \dot{q}^i(t) -H(q(t),p(t)) \right)}}.$\end{tiny}}
\end{split}
\end{equation}

This is the rigorous notion of \eqref{1} as the observable quantity is expressed as a weighted sum over paths in $\Gamma.$ We will often abbreviate the left-hand side as $\langle O_1 (t_1) \cdots O_n(t_n) \rangle.$ Standard textbooks \cite{Peskin:1995ev} cover a version of \eqref{2} with the restriction of only position-dependent operators $O_i$ and the action appearing in \eqref{2} instead of the Hamiltonian. For us, the general case will be crucial. A proof of the general formula \eqref{2} together with an explanation of ingredients is gathered in Appendix \ref{Appendix A}. 

As can be seen from this derivation, \eqref{2} is equivalent to the Schrödinger equation. It is the incarnation of \eqref{1} as the sum of possibilities. Then, following our task mentioned in chapter \ref{Kapitel 1}, we have to rewrite \eqref{2} in a manifestly covariant way. 

The prescription \eqref{2} refers explicitly to the canonical coordinates $q$ and $p.$ To write it in a coordinate-invariant way, we introduce the $1$-form $\Theta = \sum_{i=1}^N p_i \delta q_i$ on $\Gamma.$ The exterior derivative on $\Gamma$ is in the following denoted by $\delta.$ From now on, we reserve the symbol $d$ for exterior derivatives on spacetime. It then follows that the exterior derivative $\Omega = \delta \Theta$ is the symplectic form on the phase space $\Gamma.$ Written in general coordinates $\phi^A$ for $A=1,\ldots,2N$ on $\Gamma,$ the components $\Omega_{AB}$ of $\Omega = \frac{1}{2}\Omega_{AB} \delta \phi^A \wedge \delta \phi^B$ are non-degenerate with the inverse denoted, as usual, by $\Omega^{AB}.$ For observables $f,g : \Gamma \longrightarrow \mathbb{R}$ the Poisson-bracket is defined by $\{f,g\} = \Omega^{AB} \partial_A f \partial_B g.$ 

The measure $\mathcal{D}q(t) \mathcal{D}p(t)$ on each time-slice $t$ in \eqref{2} defines a measure on $\Gamma.$ This measure can be interpreted as a volume form $\text{Vol}$ on $\Gamma.$ In the functional integration \eqref{2} the various time-slices are then connected by exterior products. With the definition of $\mathcal{D}q(t) \mathcal{D}p(t)$ as in Appendix \ref{Appendix A} the volume form can be written in a coordinate-invariant way as

\begin{equation} \label{3}
\text{Vol} := \bigwedge_{i=1}^N \frac{\delta p_i \wedge \delta q_i}{2\pi \hbar} = \frac{1}{N!} \frac{1}{(2\pi \hbar)^N} \underbrace{\Omega \wedge \ldots \wedge \Omega}_{N\text{-times}}.
\end{equation}

We then can write \eqref{2} manifestly covariant as

\begin{equation} \label{4}
\begin{split}
&\langle \Omega | TO_1^H (t_1) \cdots O_n^H(t_n) |\Omega \rangle \\
&=\lim_{T \to \infty(1-i\varepsilon)} \frac{\int \text{Vol}(t) O_1(\phi(t_1)) \cdots O_n(\phi(t_n)) 
e^{\frac{i}{\hbar} \int_{-T}^T dt \left(\Theta_{\phi(t)}[\dot \phi(t)] -H[\phi(t)] \right)}}
{\int \text{Vol}(t) 
e^{\frac{i}{\hbar} \int_{-T}^T dt \left(\Theta_{\phi(t)}[\dot \phi(t)] -H[\phi(t)] \right)}}
\end{split}
\end{equation}

where functional integration is over all paths $\phi : \mathbb{R} \longrightarrow \Gamma$ in $\Gamma.$ The coorelation functions in \eqref{4} are determined by the generating functional associated to the partition function compactly written as

\begin{equation} \label{5}
Z= \lim_{T \to \infty(1-i\varepsilon)} \int \text{Vol}(t) 
e^{\frac{i}{\hbar} \int_{-T}^T dt \left(\Theta_{\phi(t)}[\dot \phi(t)] -H[\phi(t)] \right)}.
\end{equation}

The prescription \eqref{4} or \eqref{5} is the main result of this chapter. We learn that specification of a theory requires choosing the triple $(\Gamma,\Theta,H).$ Hereby, $\Gamma$ is a manifold, $\Theta$ a $1$-form on $\Gamma$ such that $\Omega =\delta \Theta$ is non-degenerate at each point. $H:\Gamma \longrightarrow \mathbb{R}$ is an observable playing the role of the Hamiltonian. From now on, we will call such a triple $(\Gamma,\Theta,H)$ a \emph{theory}. 

With this definition, $(\Gamma,\Omega)$ is always a symplectic manifold. By Darboux's theorem, there are at least locally always canonical coordinates such that $\Theta=\sum_{i=1}^N p_i \delta q_i.$ In canonical coordinates, \eqref{4} reduces to \eqref{2}.

With \eqref{4} at hand, we could now proceed the materialization of the ideas described in the last chapter and discuss the possifold notion. Before doing that, however, we want to connect our notion of a theory with the conventional ones. As a by-product, we rederive what is known as the covariant phase space formalism. Originally intended to provide a notion of conserved charges in diffeomorphism invariant theories, we here see its origin directly from the Schrödinger equation in the form \eqref{4}. There are many reviews on this subject (e.g. \cite{Seraj:2016cym} contains a quick introduction with further references) although we have tried to be self-contained. 

In \eqref{4}, a sum over all paths in the manifold $\Gamma$ is taken (see Fig. \ref{Fig. 2} (a)).

\begin{figure}[h!]
  \centering
  \subfloat[][]{\includegraphics[width=0.4\linewidth]{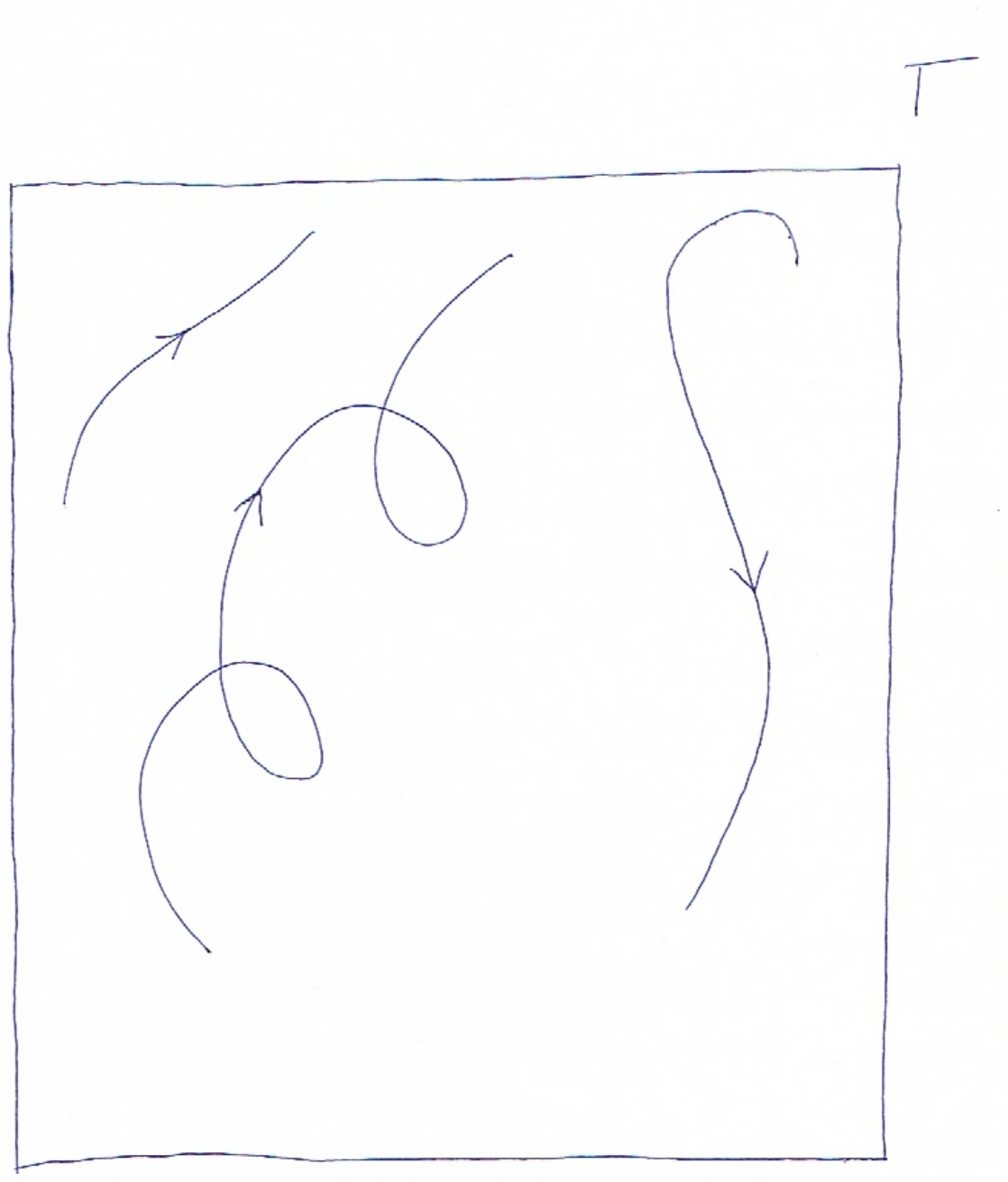}}%
  \qquad
  \subfloat[][]{\includegraphics[width=0.4\linewidth]{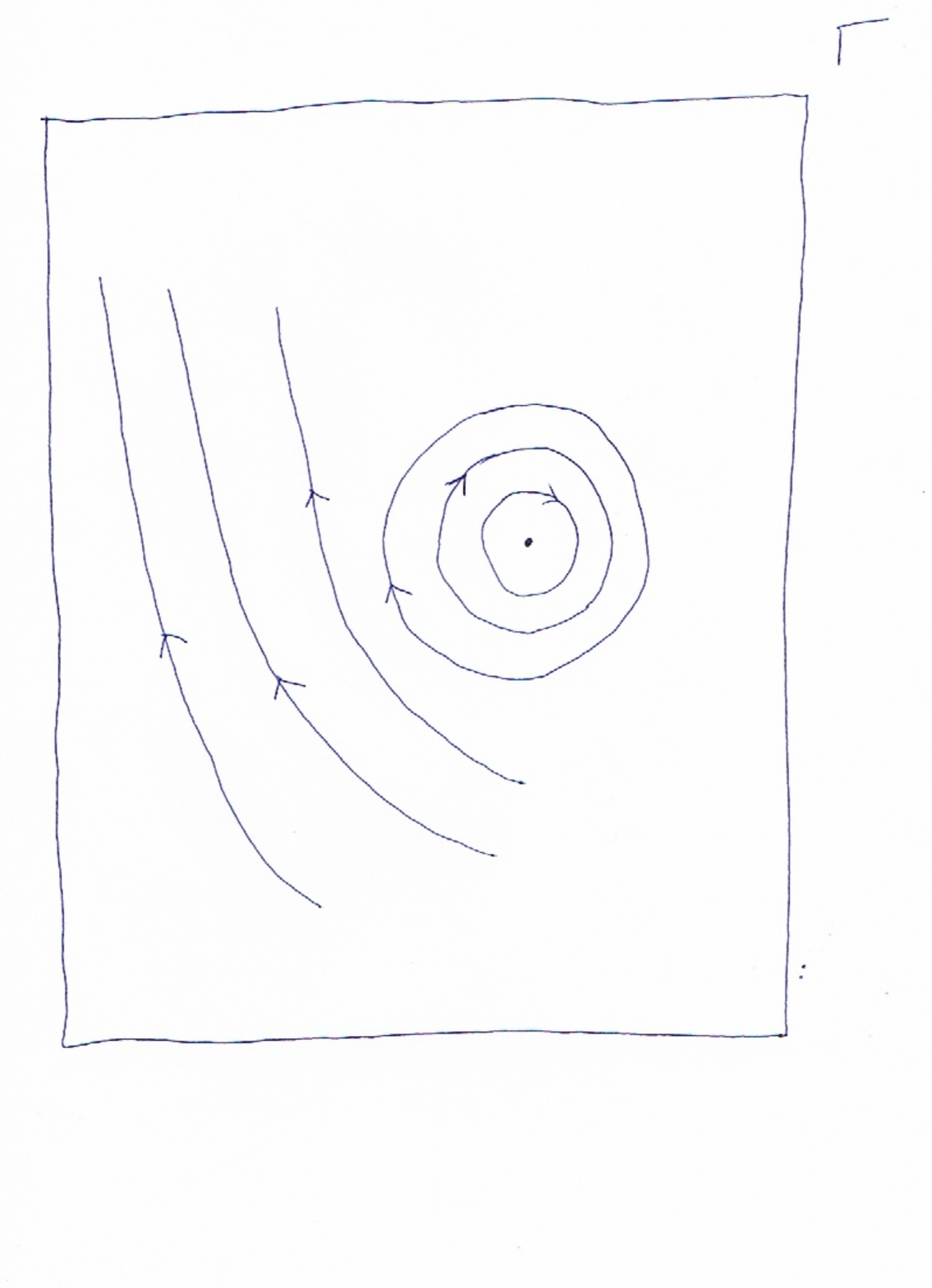}}%
  \caption{(a) Quantum mechanics insists in taking a sum over all paths in $\Gamma.$ (b) The dominant paths in the $\hbar \to 0$ limit form the phase portrait.}%
	\label{Fig. 2}
\end{figure} 

At first, it might be striking that the weight of each path contains the potential $\Theta$ of the symplectic form rather than the symplectic form itself. However, it is easily checked with the methods of Appendix \ref{Appendix A} that a change $\Theta \to \Theta + \delta a$ for an arbitrary scalar $a: \Gamma \longrightarrow \mathbb{R}$ does not affect \eqref{4} as long as the operator insertions are at finite times $t_i.$ The reason is that $a$ enters in the nominator and denominator only at the times $\pm T$ which are sent to infinity where they cancel.\footnote{Sometimes, the operator insertion times are taken to infinity. This is, for instance, the case in the derivation of the LSZ-formula. This is not problematic, however, as the limits are taken after the limit $T \to \infty$ in \eqref{4}.}

Classical mechanics is obtained as the $\hbar \to 0$ limit in \eqref{4}. The variation of the phase factor under changing the path is obtained as 

\begin{equation} \label{6}
\begin{split}
&\delta_{\phi(t)} \int dt \left( \Theta_{\phi(t)}[\dot \phi(t)] - H[\phi(t)] \right) \\
&= \int dt \left( \Omega_{AB} (\phi(t)) \dot{\phi}^B(t) - \partial_A H (\phi(t)) \right) \delta \phi^A(t) \\
&+ \int dt \frac{d}{dt} \left( \Theta_{\Phi(t)}(\delta \phi(t)) \right).
\end{split}
\end{equation}

By the same argument as in the previous paragraph, the total derivative term in \eqref{6} cancels in the nominator and denominator in \eqref{4}. The dominant paths in \eqref{4} are then found by stationary phase approximation and fulfill the Hamiltonian equations of motion

\begin{equation} \label{7}
\Omega_{AB} (\phi(t)) \dot{\phi}^B(t) = \partial_A H (\phi(t)).
\end{equation}

Hence, we see how the variational principle of classical mechanics is a consequence of \eqref{4} in the $\hbar \to 0$ limit. The dominant paths in $\Gamma$ form the phase portrait in the Hamiltonian phase space (see Fig. \ref{Fig. 2} (b)). 

The Hamilton equations of motion \eqref{7} state that the Hamiltonian is the generator of time evolution. More generally, we define an observable $H_X : \Gamma \longrightarrow \mathbb{R}$ to be a generator of the vectorfield $X$ on $\Gamma$ if $\partial_A H_X = X^B \Omega_{BA}.$ A vectorfield $X$ on $\Gamma$ is called symplectic symmetry if $\mathcal{L}_X \Omega =0.$ Because of $0=\mathcal{L}_X \Omega = \delta(X \cdot \Omega),$ there is (at least locally) a generator $H_X$ of $X.$ For $\phi \in \Gamma,$ the value $H_X[\phi]$ is called its charge with respect to $X.$ 

Note that there is a one-to-one correspondence between the observables (up to a constant) and the symplectic symmetries they generate. While this holds trivially for all maps $\Gamma \longrightarrow \mathbb{R}$ for finite-dimensional phase spaces, this restriction is called Regge-Teitelboim differentiability in the infinite-dimensional case. A map $\Gamma \longrightarrow \mathbb{R}$ is only an observable if it is the generator of a symplectic symmetry. 

One shows easily that the symplectic symmetries form a Lie-algebra with respect to the Lie-bracket on $\Gamma.$ With our sign convention, the associated generators form a representation of this algebra up to central extensions with respect to the Poisson-bracket. That is, $\{ H_X, H_Y \} = H_{[X,Y]} + K_{X,Y}$ for symplectic symmetries, their associated generators and constant $K_{X,Y} \in \mathbb{R}.$ Quantum mechanically, the operator algebra is dictated by \eqref{4}.

\subsection{Partition Functions, Action Integrals, Boundary Terms}
\label{Kapitel 2.2}

We have so far carefully described how a choice $(\Gamma,\Theta,H)$ specifies a theory and its measurable quantities. In the usual approach, a theory is given by an action. How does this relate to our notion?

We can write \eqref{5} in canonical coordinates and integrate over the momenta to yield an expression of the form

\begin{equation} \label{8}
\begin{split}
Z &= \int_{\mathcal{F}} \mathcal{D}q(t) R[q(t)] e^{\frac{i}{\hbar} S[q(t)]} \\
&= \int_{\mathcal{F}} \widetilde{\mathcal{D}q(t)} e^{\frac{i}{\hbar} S[q(t)]}
\end{split}
\end{equation}

with real functionals $R$ and $S$ (the former can be absorbed in the measure). This is the standard form for the partition function depending on the choice of an action $S[q(t)]$ (together with a measure and the configuration space $\mathcal{F} = \{q(t)\}$). 

In the classical $\hbar \to 0$ limit, the integration over momenta $p(t)$ can be estimated by stationary phase approximation. Following \eqref{6} in canonical coordinates, $p(t)$ is determined by the relation among $p(t)$ and $q(t)$ as well as $\dot q(t)$ implied by the equations \eqref{7}. From \eqref{5}, we then further obtain for $\hbar \to 0$ 

\begin{equation} \label{9}
S[q(t)] \xrightarrow[\hbar \to 0]{} \int dt \left( \sum_i p_i(q,\dot q) \dot{q}_i - H(q,p(q,\dot q)) \right),
\end{equation}

that is, we have the standard relation between the classical action and the Hamiltonian formulation (given by $(\Gamma,\Theta,H)$).

The space of functional integration in \eqref{5} is the phase space $\Gamma.$ It consists of the solutions of the Hamiltonian equations \eqref{7}. Equivalently, it is parametrized by those field configurations in $\mathcal{F}$ satisfying the equations of motion implied by the classical action \eqref{9}. In short, we have

\begin{equation} \label{10}
\Gamma \cong \bar{\mathcal{F}} := \{q(t) | q(t) \text{ satisfies equations of motion associated to \eqref{9}} \}.
\end{equation}

We want to briefly discuss the proper transition between the Lagrangian and Hamiltonian formulation as some points of this will become important later in the discussion. We consider the situation depicted in Fig. \ref{Fig. 3}.  

\begin{figure}[h!]
\centering
  \includegraphics[trim = 0mm 100mm 20mm 5mm, clip, width=0.7\linewidth]{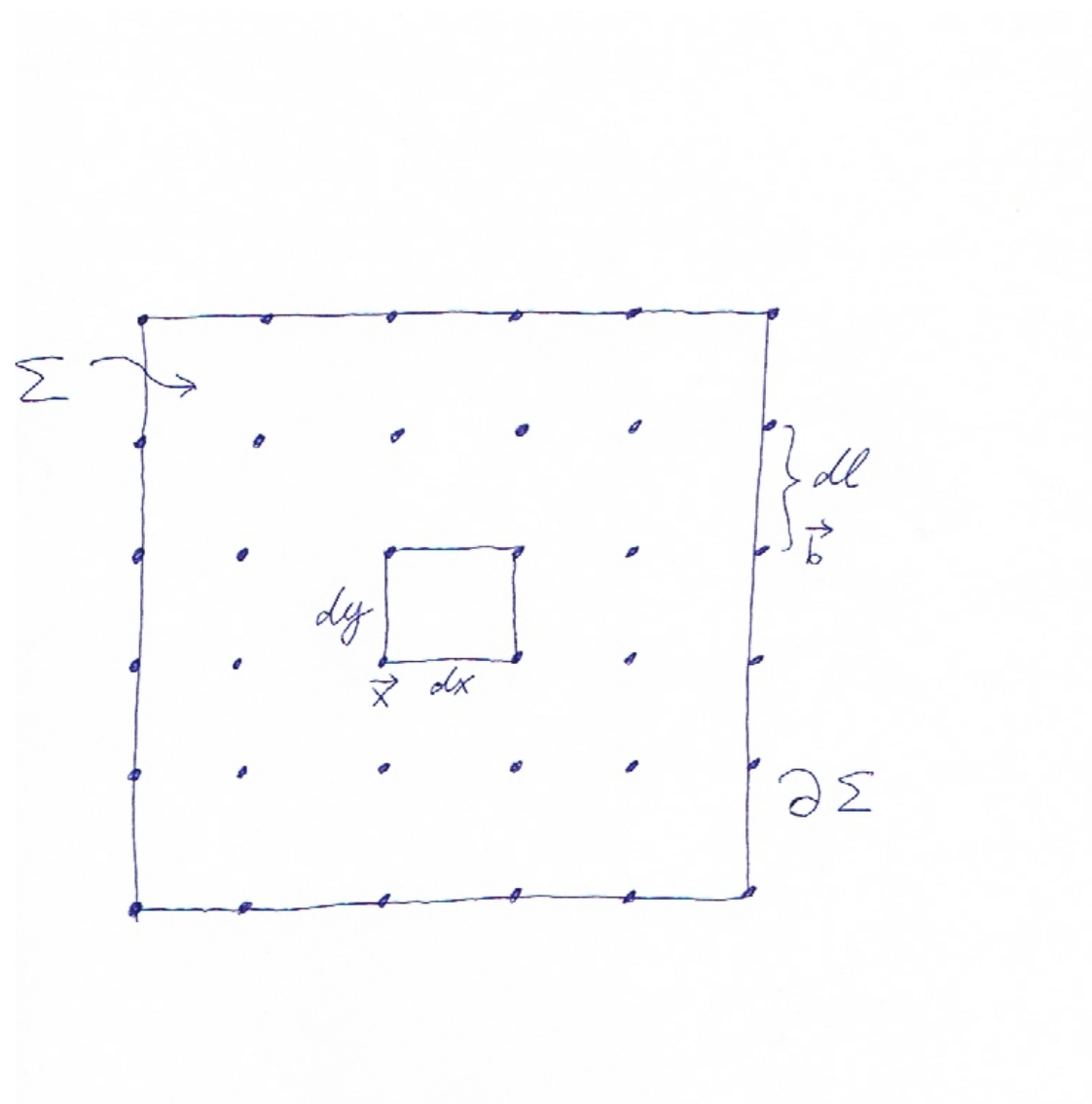}
  \caption{A system of generalized coordinates forming a hypersurface $\Sigma$ with boundary $\partial \Sigma.$}
	\label{Fig. 3}
\end{figure}

We have a system with different generalized coordinates labeled by the sites in Fig. \ref{Fig. 3}. Sites in the interior of the lattice are labeled by their respective coordinates $\x$ and are associated the coordinate $q_{\x}.$ The same holds analogously for sites $\y$ at the boundary. The system will have an action

\begin{equation} \label{11}
S=\int dt \sum_{\x} dx dy \mathcal{L}_{\x} (q_{\x},\dot{q}_{\x},\ldots) + \int dt \sum_{\y} dl \mathcal{L}_{\y} (q_{\y},\dot{q}_{\y},\ldots).
\end{equation} 

The dots in \eqref{11} describe typically nearest-neighbour interactions among the coordinates at the lattice sites although this will be not important in what follows.

In the continuum limit $dx, dy, dl \to 0,$ we can interpret the lattice as a surface $\Sigma$ forming the foliation along the evolution with $t$ (which we call ``time'') of the manifold $M = \Sigma \times \mathbb{R}$ (which we call ``spacetime''). The generalized coordinates are labeled by points in the hypersurface $\Sigma$ (with boundary $\partial \Sigma$) and evolve in time $t.$ The continuum version of \eqref{11} is

\begin{equation} \label{12}
S= \int_M L + \oint_{\partial M} B
\end{equation}

with some Lagrange-form $L$ on $M$ and an appropriate form $B$ on $\partial M.$ 

The canonical $1$-form for \eqref{11} is then as usual

\begin{equation} \label{13}
\begin{split}
\Theta &= \sum_i p_i \delta q_i \\
&= \sum_\x dx dy \frac{\partial \mathcal{L}_\x}{\partial \dot{q}_\x} \delta q_\x + \sum_\y dl \frac{\partial \mathcal{L}_\y}{\partial \dot{q}_\y} \delta q_\y.
\end{split}
\end{equation}

Practically, the interaction in the interior of $\Sigma$ is known, i.e. the Lagrange-form $L$ in \eqref{12} is known. Suppose it to be of the form $L = \varepsilon \mathcal{L}(\phi,\partial_\mu \phi),$ where we have collected $\phi(t,\x):=q_\x(t)$ the coordinates into a field. $\mathcal{L}$ is the Lagrangian density and $\varepsilon$ the Levi-Civita tensor. For the variation, we obtain

\begin{equation} \label{14}
\begin{split}
\delta L &= \varepsilon \left( \frac{\partial \mathcal{L}}{\partial \phi} -\partial_{\mu} \left( \frac{\partial \mathcal{L}}{\partial(\partial_\mu \phi)} \right) \right) \delta \phi \\
&+ \varepsilon \partial_\mu \left( \frac{\partial \mathcal{L}}{\partial(\partial_\mu \phi)} \delta \phi \right).
\end{split}
\end{equation}

With respect to the metric used to raise and lower the indexes, the second term is a total derivative $d \theta,$ where $\theta = \ast(j_\mu dx^\mu)$ with $j^\mu = \frac{\partial \mathcal{L}}{\partial(\partial_\mu \phi)} \delta \phi.$ Furthermore, we see that the interior part of \eqref{13} corresponds to 

\begin{equation} \label{15}         
\sum_\x dx dy \frac{\partial \mathcal{L}_\x}{\partial \dot{q}_\x} \delta q_\x = \int_\Sigma \theta.
\end{equation}

These observations are summarized in the basic formulas of the covariant phase space formalism. For general fields collectively denoted by $\phi$ on a $d$-dimensional spacetime $M,$ the variation of a Lagrangian-form $L[\phi]$ of degree $d$ is of the form $\delta L = -E[\phi] \delta \phi + d \theta.$ Hereby, $E[\phi]$ is a $d$-form vanishing if $\phi$ satisfies the equations of motion. The form $\theta=\theta_\phi [\delta \phi]$ of degree $d-1$ obviously can be shifted as $\theta \to \theta + dk$ for a form $k = k_\phi [\delta \phi]$ of degree $d-2.$ The form $\theta$ is called the presymplectic potential. Although often ignored in the literature, the shift by $k$ is not ambiguous. As in \eqref{15}, the canonical $1$-form is given by

\begin{equation} \label{16}
\Theta = \int_\Sigma \theta
\end{equation}

where $\Sigma$ is a hypersurface allowing a foliation of $M.$ A shift by $k$ contributes a boundary integral over $\partial \Sigma.$ Hence, we learn that the form $k$ controls the second term in \eqref{13} coming from the boundary term in the action \eqref{12}. Generically, a change in $k,$ therefore, means a change in the boundary term of \eqref{12} and hence corresponds to a different theory. This is the (often supressed) meaning of $k.$

Only if the Hamiltonian generating the time evolution along the $\Sigma$-foliation is generated with \eqref{16} guarants, as explained, the well-posedness of the initial-value problem as well as the variational principle. The choice of $k,$ explicitly entering \eqref{16}, reflects the boundary term of the action governing the degrees of freedom $q_\y (t)$ at the boundary $\partial \Sigma.$

In \eqref{16}, only the value of $k$ at $\partial \Sigma$ enters. This seems to suggest that only the boundary value of $k$ is meaningful. Strictly speaking, this is correct. However, we will give a meaning to $k$ in the interior in the context of the possifold flow in the next chapter. 

\subsection{Redundancies, Degeneracies, Faddeev-Popov Ghosts}
\label{Kapitel 2.3}

We have now discussed how to translate a theory given by an action into a formulation $(\Gamma,\Theta,H)$ presented here. While the latter will be central for our discussion, the former is the usual way a theory is presented. The reason is that to find a suited theory to describe a physical phenomenon, the guiding principles are the symmetries of this phenomenon. Within a Lagrangian formulation, those symmetries are typically manifest. This restricts the way the action of a searched theory can be like. 

This point, however, is well-known to be connected with some caveats which will be important for us. Consider, for instance, the case of electrodynamics. Its action principle in terms of locally interacting fields on spacetime with manifest Lorentz-invariance contains redundant degrees of freedom. The example means that for an action $S[q(t)]$ it may happen that not all generalized coordinates $q$ are physical. Rather, some of them may correspond to redundancies in the description. As explained, these redundancies typically appear to make particular symmetries of the action manifest in terms of locally interacting spacetime fields.

These redundancies need to be excluded from the sum \eqref{8}. Practically, this happens through the introduction of additional locally interacting spacetime fields. The role of these Faddeev-Popov ghosts is to cancel the redundancies at each order of perturbation theory following from \eqref{8}. While this standard procedure has been proven very useful in practical computations as it allows the use of conventional Feynman diagrams, the high amount of present redundancy is evident.

As shown, any system following the Schrödinger equation must possess a formulation in the form of \eqref{5}. \eqref{5} is free of redundancies and we want to ask how the transition to the formulation \eqref{5} is getting rid of the redundancies present in the action. 

In order to see this, suppose a system to be described by the action $S=\int dt L(q, \dot q)$ with a Lagrange-function $L=L(q,\dot q)$ and generalized coordinates $q=(q_1,\ldots,q_N).$ Using the relation $p_i = \frac{\partial L}{\partial \dot{q}_i}$ among canonical momenta and velocities, we can express the symplectic form on the phase space in terms of coordinates $q$ and velocities $\dot q$ as

\begin{equation} \label{17}
\begin{split}
\Omega &= \delta \Theta = \delta \left( \sum_i p_i \delta q_i \right) \\
&= \sum_{i,j} \frac{\partial^2 L}{\partial \dot{q}_i \partial q_j} \delta q_j \wedge \delta q_i 
+\sum_{i,j} \frac{\partial^2 L}{\partial \dot{q}_i \partial \dot{q}_j} \delta \dot{q}_j \wedge \delta q_i.
\end{split}
\end{equation}

This implies that at any point $\Omega$ is non-degenerate if and only if the Hesse-matrix $\left( \frac{\partial^2 L}{\partial \dot{q}_i \partial \dot{q}_j} \right)_{i,j=1,\ldots,N}$ is non-degenerate. According to the implicit function theorem, the degeneracy of this Hesse-matrix means that the relation $p_i=\frac{\partial L}{\partial \dot{q}_i}$ is not solvable for the velocities $\dot q.$ This signals the presence of gauge constraints and redundancies. Hence, by dividing out the degenerate directions of $\Omega$ at any point, one is getting rid of the redundancies directly.

In total, the transition from the action $S[q(t)]$ to the formulation $(\Gamma,\Theta,H)$ proceeds as described where $\Gamma$ is determined by \eqref{10} with the addendum that the right-hand-side has to be divided out by the zero-modes of $\Omega=\delta \Theta$ given by \eqref{16}.

In this chapter, we have given with \eqref{4} and \eqref{5} an explicit realization of the sum of possibilities making covariance manifest. We have addressed how this formulation of a theory relates to the conventional one using the action. Hereby, we have rederived the covariant phase space formalism and especially emphasized the appearance of boundary terms often not properly treated in the literature. The next question is clear. How can we evaluate \eqref{4} for a given theory? Or what are at least tools to simplify this task? This leads directly to the possifold notion.

\section{Possifolds}
\label{Kapitel 3}

In this chapter, let $(\Gamma,\Theta,H)$ denote a given theory in the sense of chapter \ref{Kapitel 2}. The measurable quantities are then given by \eqref{4}. Unfortunately, performing the sum \eqref{4} is in general very difficult. Already for finitely many degrees of freedom and especially in quantum field theory one needs to resort to suited approximations (e.g. perturbation theory). As mentioned in chapter \ref{Kapitel 1}, we want to discuss here a further such approximation method. The idea is to divide the sum for the theory $(\Gamma,\Theta,H)$ into suited subsums. 

The division into these subsums is natural and of physical meaning. To understand this, suppose we are placing an oscillator on a table and we are interested in performing measurements. We are typically not interested in all possible measurable quantities. It may be sufficient for us to measure the oscillator's position and velocity while we are not interested in the electromagnetic field in the vicinity of the oscillator or for the gravitational field at spatial infinity (although both may affect the oscillator). 

Concretely, this means that in \eqref{4}, we are typically interested in observables $O : \Gamma \longrightarrow \mathbb{R}$ that only partially depend on the coordinates of $\Gamma.$ Suppose, $\Gamma = S \times T$ for manifolds $S$ and $T$ and we consider observables $O : \Gamma \to \mathbb{R}$ that are independent of the coordinates on $T.$ On each time-slice in \eqref{4}, we then would like to split the integration over $\Gamma$ into a product integration over the coordinates on $S$ and $T,$ respectively. Since the operator insertions are independent of the $T$-coordinates, we would like to integrate them out. The hope is that the resulting summation over paths in $S$ is again of the form \eqref{4}. This leads to the following definition:

\begin{definition} \label{Definition 3.1}
Suppose, $\Gamma=S\times T$ for two manifolds $S$ and $T.$ $S$ is called \emph{possifold}, if there is a $\lambda \in \mathbb{C}$ with

\begin{equation} \label{18}
\begin{split}
& \int \emph{Vol}(t) 
e^{\frac{i}{\hbar} \int dt \left(\Theta_{\phi(t)}[\dot \phi(t)] -H[\phi(t)] \right)} \left( \cdots \right) \\
&= \lambda  \int_{\phi(t) \in S} \emph{Vol}'(t) 
e^{\frac{i}{\hbar} \int dt \left(\Theta^{\prime}_{\phi(t)}[\dot \phi(t)] -H'[\phi(t)] \right)} \left( \cdots \right)
\end{split}
\end{equation}

for a suited potential $\Theta'$ of a symplectic form on $S$ and a Hamiltonian $H':S \longrightarrow \mathbb{R}.$ $\emph{Vol}'$ denotes the induced volume form. The summation on the right-hand-side is over all paths in $S.$ The dots $(\cdots)$ on both sides denote an arbitrary expression depending for each time-slice only on $S.$ 
\end{definition}

Several remarks are in order. 

$S$ can be thought of as foliating $\Gamma$ along $T.$ Therefore, $S$ can be thought of as a submani\emph{fold} in the space of \emph{possi}bilities $\Gamma.$ Hence, the name possifold is used.

The conditions on $S$ in definition \ref{Definition 3.1} are non-trivial. A simple counterexample would be an odd-dimensional $S.$ $S$ is not a possifold due to the lack of a symplectic form.

A possifold gives via the triple $(S,\Theta',H')$ rise to a theory on its own right (in the sense of chapter \ref{Kapitel 2}). From \eqref{18}, the associated partition function \eqref{5} is

\begin{equation} \label{19}
Z' = \int_{\phi(t) \in S} \text{Vol}'(t) 
e^{\frac{i}{\hbar} \int dt \left(\Theta^{\prime}_{\phi(t)}[\dot \phi(t)] -H'[\phi(t)] \right)}.
\end{equation}

The observables $O': S \longrightarrow \mathbb{R}$ on $S$ can be mapped via $O(s,t)=O'(s)$ to the $T$-independent observables $O:\Gamma \longrightarrow \mathbb{R}$ on $\Gamma.$ Due to \eqref{18}, we have for the correlation functions of such related observables

\begin{equation} \label{20}
\langle O_1(t_1) \cdots O_n(t_n) \rangle_{Z}
= \langle O^{\prime}_1 (t_1) \cdots O^{\prime}_n (t_n) \rangle_{Z'}
\end{equation} 

where each side is computed with respect to the theory as indicated by the partition function in the subscript.

From a mathematical point of view the possifolds are the natural substructures of a given theory. This substructures are sensitive to the theory they are embedded in. 

Note that definition \ref{Definition 3.1} is the strictest notion of a possifold. In practice, we will be not that strict. While \eqref{18} requires the integration over the whole manifold $T,$ we will usually approximate this. Consequently, in equations \eqref{18}-\eqref{20}, the equality is replaced by a $\approx.$ The degree of approximation has to be specified and justified by how much of the $T$-integration is taken into account. Quite often, the quantum mechanical corrections due to $T$-integration are negligible and it will be sufficient to ignore them. 

The following simple example will occur quite often. It will have an interesting realization in field theory.

\begin{example} \label{Example 3.1}
Suppose, $\Gamma = \{(q,p) \} = \{ (q_1,\ldots,q_N,p_1,\ldots,p_N) \}$ to be parametrized by a global canonical chart. For a suited $k \in \mathbb{N},$ we can write $\Gamma = S \times T$ with $S=\{ (q_1,\ldots,q_k,p_1,\ldots,p_k) \}$ and
$T=\{(q_{k+1},\ldots,q_N,p_{k+1},\\\ldots,p_N)\}.$ We then have the inherited canonical potentials $\Theta' = \sum_{i=1}^k p_i \delta q_i$ on $S$ and $\Theta'' = \sum_{i=k+1}^N p_i \delta q_i$ on $T.$ The Hilbert space factorizes accordingly $\mathcal{H} = \mathcal{H}_S \otimes \mathcal{H}_T.$ If $S$ and $T$ are decoupled, i.e. the Hamiltonian decomposes as $H=H'+H''=H'(q_1,\ldots,p_k)+H''(q_{k+1},\ldots,p_N),$ then $S$ is a possifold. 

Indeed, \eqref{18} holds for the triple $(S,\Theta',H')$ with

\begin{equation} \label{21}
\lambda= \int \text{Vol}''(t) 
e^{\frac{i}{\hbar} \int dt \left(\Theta^{\prime \prime}_{\phi(t)}[\dot \phi(t)] -H''[\phi(t)] \right)}.
\end{equation}

This shows $S$ to be a possifold in the strictest sense. The notion is only approximate or is satisfied in a certain limit if the decoupling condition is only approximate or is satisfied in a certain limit.      
\end{example} 

In the remainder of this chapter, we want to collect some basic facts on possifolds to be used occasionally in applications.

\subsection{Basic Facts on Possifolds}
\label{Kapitel 3.1}

The following assertion states that each possifold can always be thought of as being characterized by a set of charges forming a closed algebra. 

\begin{lemma} \label{Lemma 3.1.1}
Let $S$ be a possifold. Locally, $S$ can be parametrized by the charges of generators of a closed algebra of symplectic symmetries.
\end{lemma}

\begin{proof}
Locally, $S$ has a chart of canonical coordinates $(\tilde{q}_i,\tilde{p}_i)$ for $i=1,\ldots,\\
 \frac{1}{2} \dim S.$ The $\tilde{q}_i$ and $\tilde{p}_i$ are generators of a closed algebra of symplectic symmetries.  
\end{proof}

Practically, Lemma \ref{Lemma 3.1.1} can be used in the search for possifolds in a given theory. We have to seek for charges forming a closed algebra. We want to make this more explicit. 

\subsubsection{The Possifold-Flow}
\label{Kapitel 3.1.1}

We assume the situation as in example \ref{Example 3.1} where the manifold $S$ is taken to be $\Gamma^{(k)} := S=\{ (q_1,\ldots,q_k,p_1,\ldots,p_k) \}$ depending on the choice of $k \in \mathbb{N}$ with the canonical potential $\Theta'=\Theta^{(k)}$ possessing the symplectic form $\Omega^{(k)} = \sum_{i=1}^k \delta p_i \wedge \delta q_i$ on $\Gamma^{(k)}.$ 

We then have the following

\begin{lemma} \label{Lemma 3.1.2}
Let $X^A=(X^q,X^p)$ be a symplectic symmetry in $(\Gamma,\Omega).$ For each fixed $(q_{k+1},\ldots,q_N,p_{k+1},\ldots,p_N)$ is then $X^{(k)}=(X^{q_1},\ldots,X^{q_k},X^{p_1},\ldots,\\ X^{p_k})$ a symplectic symmetry in $(\Gamma^{(k)},\Omega^{(k)}).$

Conversely, each symplectic symmetry $X^{(k)}$ in $(\Gamma^{(k)},\Omega^{(k)})$ defines a symplectic symmetry $X$ in $(\Gamma,\Omega)$ via

\begin{equation} \label{22}
\begin{split}
X^A &= (X^{q_1},\ldots,X^{q_k},X^{q_{k+1}},\ldots,X^{q_N},X^{p_1},\ldots,X^{p_k},X^{p_{k+1}},\ldots,X^{p_N}) \\
&= (X^{q_1},\ldots,X^{q_k},0,\ldots,0,X^{p_1},\ldots,X^{p_k},0,\ldots,0).
\end{split}
\end{equation}
\end{lemma}

\begin{proof}
For a symplectic symmetry in $(\Gamma,\Omega),$ we have $\mathcal{L}_X \Omega_{AB} =0$ which in canonical coordinates reads

\begin{equation} \label{23}
\begin{split}
\frac{\partial X^{p_i}}{\partial q_j} - \frac{\partial X^{p_j}}{\partial q_i} &= 0 \\
\frac{\partial X^{p_i}}{\partial p_j} + \frac{\partial X^{q_j}}{\partial q_i} &= 0 \\
\frac{\partial X^{q_i}}{\partial p_j} - \frac{\partial X^{q_j}}{\partial p_i} &= 0
\end{split}
\end{equation}

for all $i,j=1,\ldots,N.$ Hence, for each fixed $(q_{k+1},\ldots,q_N,p_{k+1},\ldots,p_N),$ the equations hold for all $i,j=1,\ldots,k.$ This is the condition for $X^{(k)}$ to be a symplectic symmetry in $(\Gamma^{(k)},\Omega^{(k)}).$ 

Conversely, \eqref{22} fulfills \eqref{23}. For $i,j=1,\ldots,k$ this follows from the requirement on $X^{(k)}.$ The remaining conditions are empty. 
\end{proof}

We look at the relation among the generators of the symplectic symmetries in Lemma \ref{Lemma 3.1.2}. While $A,B,\ldots$ denote the indexes of the coordinates in $\Gamma,$ let $\tilde A, \tilde B,\ldots$ denote the ones on $\Gamma^{(k)}.$ Lemma \ref{Lemma 3.1.2} associates to a symplectic symmetry $X$ on $(\Gamma,\Omega)$ the symplectic symmetry $X^{(k)}$ on $(\Gamma^{(k)},\Omega^{(k)}).$ The corresponding generator is determined by $\partial_{\tilde A} H_{X^{(k)}} = X^{(k) \tilde B} \Omega^{(k)}_{\tilde B \tilde A}$ and hence by

\begin{equation} \label{24}
\begin{split}
\delta H_{X^{(k)}} &= \partial_{\tilde A} H_{X^{(k)}} \delta \phi^{\tilde A} \\
&= X^{(k) \tilde B} \Omega^{(k)}_{\tilde B \tilde A} \delta \phi^{\tilde A} = \Omega^{(k)}(X^{(k)},\delta \phi).
\end{split}
\end{equation}

For $k=N,$ $H_{X^{(k=N)}} = H_X$ is the generator of $X$ on $(\Gamma,\Omega).$ 

If conversely, $X^{(k)}$ is a symplectic symmetry on $(\Gamma^{(k)},\Omega^{(k)}),$ the associated generator $H_{X^{(k)}}$ is equally determined by \eqref{24}. Viewed as a function on $\Gamma,$ it follows that it is the generator of \eqref{22} on $(\Gamma,\Omega).$ 

The relations are summarized in the following 

\begin{lemma} \label{Lemma 3.1.3}
Let $X$ be a symplectic symmetry on $(\Gamma,\Omega).$ The generator $H_{X^{(k)}}$ of the according to Lemma \ref{Lemma 3.1.2} induced symplectic symmetry $X^{(k)}$ on $(\Gamma^{(k)},\Omega^{(k)})$ is given by \eqref{24}. 

Conversely, for a symplectic symmetry $X^{(k)}$  on $(\Gamma^{(k)},\Omega^{(k)})$ with generator $H_{X^{(k)}},$ \eqref{22} is generated on $(\Gamma,\Omega)$ by $H_{X^{(k)}}.$  
\end{lemma}

Symplectic symmetries on $(\Gamma^{(k)},\Omega^{(k)})$ induce symplectic symmetries on $(\Gamma,\Omega).$ We have clarified the relation among generators. The following relates the associated Lie-algebras with respect to the Lie-bracket on the respective manifold. 

\begin{lemma} \label{Lemma 3.1.4}
For a symplectic symmetry $X^{(k)}$ on $(\Gamma^{(k)},\Omega^{(k)}),$ let $X$ denote the associated symplectic symmetry \eqref{22} on $(\Gamma,\Omega).$ An algebra of symplectic symmetries on $(\Gamma^{(k)},\Omega^{(k)})$ induces this way an algebra of symplectic symmetries on $(\Gamma,\Omega)$ and the associated Lie-algebras are equal.  
\end{lemma}

\begin{proof}
For symplectic symmetries $X^{(k)}_a, X^{(k)}_b, X^{(k)}_c$ on $\Gamma^{(k)},$ suppose the commutator relation 

\begin{equation} \label{25}
\left[ X^{(k)}_a, X^{(k)}_b \right] = f_{abc} X^{(k)}_c
\end{equation}

with structure constants $f_{abc} \in \mathbb{R}.$ This means on $\Gamma^{(k)}$ the relation

\begin{equation*}
X^{(k) \tilde B}_a \partial_{\tilde B} X^{(k) \tilde A}_b - X^{(k) \tilde B}_b \partial_{\tilde B} X^{(k) \tilde A}_a
=f_{abc} X^{(k) \tilde A}_c.
\end{equation*}

By definition of \eqref{22}, this implies the relation on $\Gamma$ 

\begin{equation*}
X^B_a \partial_B X^A_b - X^B_b \partial_B X^A_a = f_{abc} X^A_c
\end{equation*}

which means $[X_a,X_b]=f_{abc}X_c$ on $\Gamma$ and proves the assertion. 
\end{proof}

As the generator algebra respects the Lie-algebra, this implies 

\begin{lemma} \label{Lemma 3.1.5}
Let $X^{(k)}$ and $Y^{(k)}$ be symplectic symmetries on $(\Gamma^{(k)},\Omega^{(k)})$ and $X,Y$ the associated ones on $(\Gamma,\Omega)$ in the context of Lemma \ref{Lemma 3.1.2}. On $(\Gamma^{(k)},\Omega^{(k)}),$ one has for the associated generators

\begin{equation} \label{26}
\{ H_{X^{(k)}}, H_{Y^{(k)}} \} = H_{[X^{(k)},Y^{(k)}]} + K_{X^{(k)},Y^{(k)}}.
\end{equation}

On $(\Gamma,\Omega)$ holds the relation

\begin{equation} \label{27}
\{ H_{X}, H_{Y} \} = H_{[X,Y]} + K_{X^{(k)},Y^{(k)}}
\end{equation}

with constant c-numbers $K_{X^{(k)},Y^{(k)}}.$
\end{lemma}

\begin{proof}
\eqref{26} is a general property of generators. \eqref{27} follows from \eqref{26} using Lemma \ref{Lemma 3.1.3} and \ref{Lemma 3.1.4}.
\end{proof}

The space $\Gamma^{(k)}$ describes a part of the degrees of freedom of the original space $\Gamma.$ In the previous statements of this section, we have given the relation among the symplectic symmetries and generators in both spaces. The associated maps describe the action under the reduction 

\begin{equation} \label{28}
N \to k
\end{equation}

of the $N$ degrees of freedom in $\Gamma$ to the $k$ degrees of freedom in $\Gamma^{(k)}.$ In analogy with the Wilsonian renormalization group flow reducing high energy degrees of freedom, we will denote the reduction \eqref{28} as the \emph{possifold flow}. 

The reason for this terminology is that the $\Gamma^{(k)}$ are good candidates for possifolds. In fact, in practice we will mostly look at these spaces in order to construct possifolds. Similarly to the Wilsonian notion, definition \ref{Definition 3.1} requires integrating out the reduced degrees of freedom in \eqref{28}. As explained, this integration determines to which degree of approximation $\Gamma^{(k)}$ defines a possifold. For instance, $\Gamma^{(k)}$ defines a possifold in case the decoupling condition holds as in example \ref{Example 3.1}.\\ 
\\
Finally, we want to look at the realization of the possifold flow in the field theory case. This will be the case most important for us. We consider again the situation associated with the discussion to Fig. \ref{Fig. 3}. 

From \eqref{16}, we have $\Omega=\int_\Sigma \omega$ where the presymplectic current $\omega = \delta \theta$ is a $2$-form on $\Gamma$ and a form of degree $d-1$ on $M.$ It is easily shown that $\omega$ is closed on spacetime. This follows from $d \omega = \delta(d \theta) = \delta ( \delta L + E[\phi] \delta \phi) = 0$ on $\Gamma$ as the elements of $\Gamma$ satisfy the equations of motion \eqref{10}. Let $F$ be the associated $2$-form on $\Gamma$ and $(d-2)$-form on $M$ with $\omega=dF.$ For the situation in Fig. \ref{Fig. 3}, we then have as in \eqref{13} 

\begin{equation} \label{29}
\Omega = \sum_{\x \in \Sigma} \delta p_\x \wedge \delta q_\x = \int_\Sigma \omega = \oint_{\partial \Sigma} F.
\end{equation}

Then, the choice of a $k$ in the possifold flow \eqref{28} corresponds to the choice of a subset $\Sigma(B) \subseteq \Sigma$ bounded by a codimension-2 surface $B=\partial(\Sigma(B)).$ In Fig. \ref{Fig. 3} this could correspond, for instance, to a choice as in Fig. \ref{Fig. 4}. 

\begin{figure}[h!]
\centering
  \includegraphics[trim = 0mm 100mm 20mm 5mm, clip, width=0.7\linewidth]{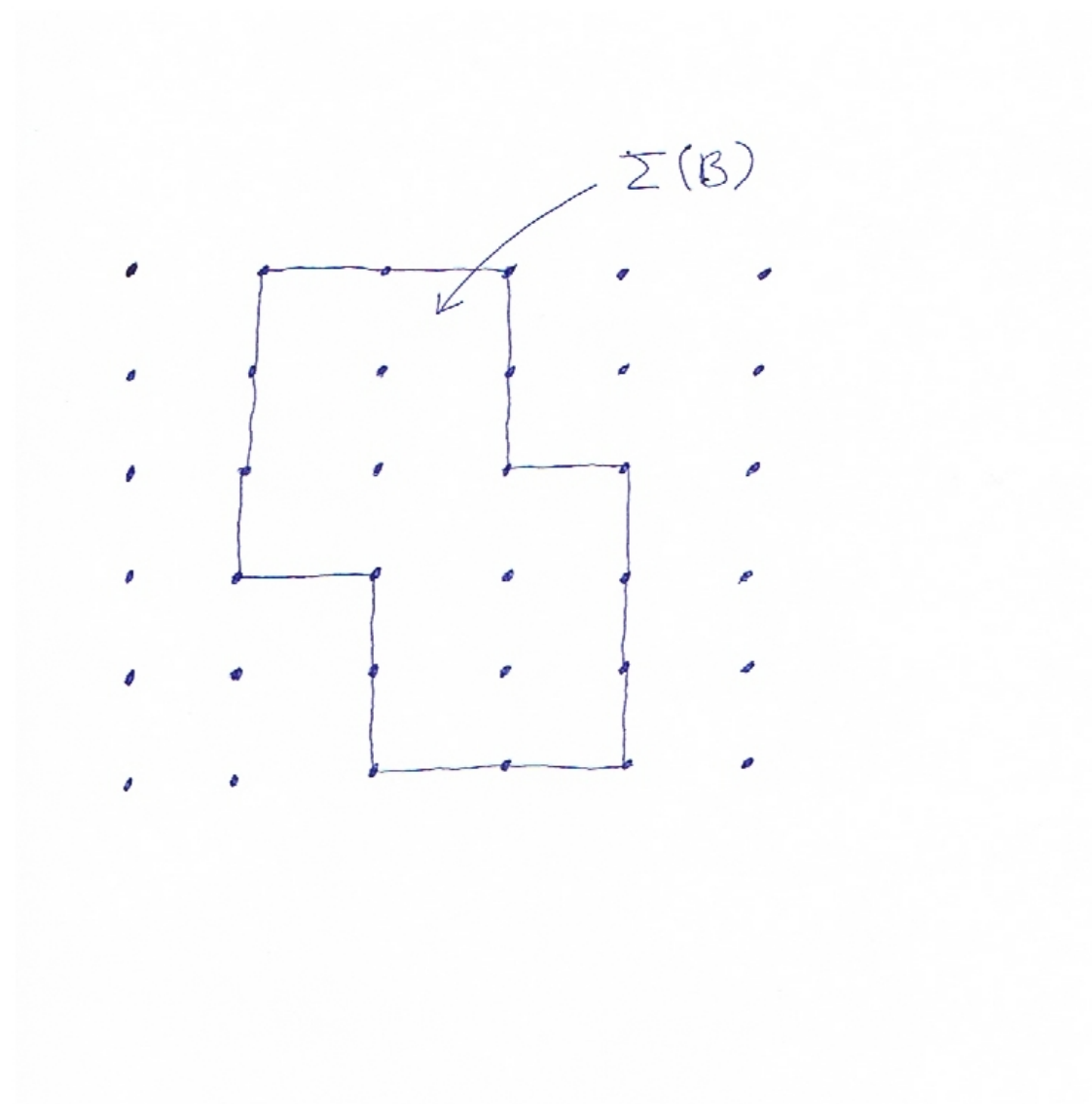}
  \caption{Possifold flow as a deformation of the boundary from $\partial \Sigma$ to $B.$}
	\label{Fig. 4}
\end{figure} 

Following the notations of this section the space $\Gamma^{(k)}$ corresponds to the space $\Gamma^{(B)}$ with the degrees of freedom covered by $\Sigma(B).$ The associated $\Omega^{(B)}$ is then 

\begin{equation} \label{30}
\Omega^{(B)} = \int_{\Sigma(B)} \omega = \oint_B F
\end{equation}

while \eqref{24} reads as

\begin{equation} \label{31}
\delta H_{X^{(B)}} = \Omega^{(B)} [X^{(B)},\delta \phi] = \int_{\Sigma(B)} \omega[X, \delta \phi] = \oint_B F[X, \delta \phi]
\end{equation}

for a symplectic symmetry $X$ on $\Gamma$ and its image $X^{(B)}$ in $\Gamma^{(B)}$ under the possifold flow.

We learn that the possifold flow \eqref{28} is in field theory realized as a deformation of the boundary

\begin{equation} \label{32}
\partial \Sigma \to B
\end{equation}

to a codimension-2 surface $B$ associated with a replacement in the corresponding formulas such as \eqref{31}.

With these observations, the meaning of the shift $\theta \to \theta + dk$ discussed in chapter \ref{Kapitel 2} becomes apparent. We have already seen that the symplectic potential \eqref{16} is sensitive to the choice of $k$ at $\partial \Sigma.$ We now understand its meaning in the interior. It controls the splitting of the degrees of freedom among $\Sigma(B)$ and its complement for the reduction \eqref{32}. The choice is arbitrary but determined by the choice of $k$ in the interior of $\Sigma.$ A sensible criterion is of course the requirement for $\Gamma^{(B)}$ to be a possifold to a high degree of approximation.\footnote{The discussion of this chapter justifies the approach of \cite{Averin:2018owq} in the analysis of black holes. \cite{Averin:2019zsi} reads as a proposal to explain the universality of black hole entropy as a consequence of the possifold flow. We, however, do not go further into this topic here.}

\section{Applications}
\label{Kapitel 4}

The discussion so far has been quite formal. The purpose of this chapter is to provide an outlook on some applications. In the next section, we illustrate how the presented formalism is used in a concrete example.

\subsection{Lieb-Liniger Model}
\label{Kapitel 4.1}

We consider a field theory with a simple interaction term. Our model is known as the Lieb-Liniger model and is discussed recently in \cite{Kanamoto_2003}. It describes non-relativistic particles with an attractive interaction. A more general review on this type of models and their terminology is in \cite{Dalfovo:1999zz}. In what follows, we use the notation of the previous chapters. We consider a field $\Psi(x,t)$ on $M=\Sigma \times \mathbb{R}$ with $\Sigma=S^1$ being parametrized by $x \in [0,2\pi].$ We take the equations of motion parametrizing $\Gamma$ as in \eqref{10} to be the Gross-Pitaevskii equation

\begin{equation} \label{33}
i \frac{\partial \Psi}{\partial t} = \left( -\Delta -g|\Psi|^2 \right) \Psi.
\end{equation}

On $\Gamma,$ we define the symplectic potential

\begin{equation} \label{34}
\Theta_\Psi [\delta \Psi] = \frac{i}{2} \int dx \Psi^\dagger \delta \Psi +h.c.
\end{equation}

where the integration is over $\Sigma \times \{t=0\}$ in all integrals to follow in this section. It induces the symplectic form

\begin{equation} \label{35}
\Omega_\Psi[\delta_1 \Psi, \delta_2 \Psi] = i \int dx \left( \delta_1 \Psi^\dagger \delta_2 \Psi - \delta_2 \Psi^\dagger \delta_1 \Psi \right).
\end{equation}

With \eqref{35} it is easily checked that the time evolution \eqref{33} is generated by the Hamiltonian

\begin{equation} \label{36}
H = \int dx \Psi^\dagger (-\Delta) \Psi - \frac{g}{2} \int dx \Psi^\dagger \Psi^\dagger \Psi \Psi.
\end{equation}

This specifies the theory $(\Gamma,\Theta,H)$ under consideration. \eqref{33} has the $U(1)$-invariance $\Psi \to e^{i \alpha} \Psi.$ Again from \eqref{35}, this is a symplectic symmetry generated by $N[\Psi]=\int dx \Psi^\dagger \Psi$ which is the particle number. 

A point $\Psi \in \Gamma$ is a solution of \eqref{33} and hence uniquely determined by $\Psi(t=0,x)$ which we can decompose on $\Sigma$ as

\begin{equation} \label{37}
\Psi(t=0,x)=\sum_{k \in \mathbb{Z}} a_k \frac{e^{ikx}}{\sqrt{2 \pi}}
\end{equation}

for unique $a_k \in \mathbb{C}.$ They provide, therefore, a parametrization of $\Gamma$ 

\begin{equation} \label{38}
\Gamma = \{a_k \in \mathbb{C} | k \in \mathbb{Z} \}.
\end{equation}

We rewrite these coordinates for $k \in \mathbb{Z} \backslash \{0\}$ as

\begin{equation} \label{39}
\begin{split}
a_k &= \frac{1}{\sqrt{2}}(q_k+ip_k) \\
q_0 &= - \operatorname{Arg} (a_0) \\
p_0 &= |a_0|^2
\end{split}
\end{equation}

with $q_k,p_k \in \mathbb{R}.$ Up to an irrelevant exact form, the symplectic potential takes in this coordinates the canonical form

\begin{equation} \label{40}
\Theta_\Psi [\delta \Psi] = p_0 \delta q_0 +\sum_{k \neq 0} p_k \delta q_k.
\end{equation}

In the parametrization \eqref{38}, the particle number is given by $N=\sum_k a_k^\dagger a_k.$ For a fixed particle number $N$ the lowest-energy state is expected to be given by $a_k=\sqrt{N} \delta_{k,0}$ for sufficiently small $g.$ The corresponding solution of \eqref{33} is

\begin{equation} \label{41}
\Phi(t,x)=\sqrt{\frac{N}{2 \pi}} e^{-i \mu t}
\end{equation}    

with $\mu = -\frac{gN}{2\pi}.$ One is now often interested in describing excitations of the Bose-Einstein condensate \eqref{41}. These excitations are given by the higher modes in \eqref{38}. Therefore, we consider the submanifold of \eqref{38} given by $q_0=0$ and $p_0=N-\sum_{k \neq 0} |a_k|^2$ on which the pull-back of \eqref{40} takes the form

\begin{equation} \label{42}
\Theta=\sum_{k \neq 0} p_k \delta q_k.
\end{equation}

Equation \eqref{42} now precisely reflects the situation in example \ref{Example 3.1}. We want to explore which of the submanifolds obtained by the possifold flow lead to sensible possifolds. 

To this end, we split the Hamiltonian

\begin{equation} \label{43}
H=H_{quadratic}+H_{higher}
\end{equation}

into a part at most quadratic in the coordinates $q_k,p_k$ and a part containing higher order terms. Since $H_{higher}$ results from the interaction term in \eqref{36}, it is suppressed by a factor of $g.$ 

We diagonalize $H_{quadratic}$ by introducing for each $k \in \mathbb{N}$ the canonical coordinates $A_k,B_k,C_k,D_k \in \mathbb{R}$ by

\begin{equation*}
\begin{split}
\begin{pmatrix}
q_k \\
q_{-k}
\end{pmatrix}
&=\frac{1}{\sqrt{2}} 
\begin{pmatrix}
1 \\
1
\end{pmatrix}
A_k+\frac{1}{\sqrt{2}}
\begin{pmatrix}
1 \\
-1
\end{pmatrix}
B_k \\
\begin{pmatrix}
p_k \\
p_{-k}
\end{pmatrix}
&= \frac{1}{\sqrt{2}}
\begin{pmatrix}
1 \\
-1
\end{pmatrix}
C_k+\frac{1}{\sqrt{2}}
\begin{pmatrix}
1 \\
1
\end{pmatrix}
D_k
\end{split}
\end{equation*}

such that \eqref{42} reads

\begin{equation} \label{44}
\Theta=\sum_{k \in \mathbb{N}} (D_k \delta A_k + C_k \delta B_k).
\end{equation}

Defining the parameter $\gamma=\frac{gN}{2\pi}$ controlling the collective interaction, we obtain

\begin{equation} \label{45}
H_{quadratic} = -\frac{\gamma N}{2} + \sum_{k \in \mathbb{N}} \left( \frac{1}{2}k^2D^2_k + (\frac{1}{2}k^2-\gamma)A^2_k+(\frac{1}{2}k^2-\gamma)C^2_k + \frac{1}{2}k^2B^2_k\right).
\end{equation}

For each $k \in \mathbb{N},$ this can be brought into the form of two non-relativistic particles by rescaling

\begin{equation} \label{46}
\begin{split}
P_k &= kD_k \\
Q_k &= \frac{1}{k}A_k \\
\tilde{P}_k &= (k^2-2\gamma)^{\frac{1}{2}} C_k \\
\tilde{Q}_k &= (k^2-2\gamma)^{-\frac{1}{2}} B_k.
\end{split}
\end{equation}

This coordinates remain canonical with the term in the brackets in \eqref{45} becoming

\begin{equation*}
\frac{1}{2}P^2_k+\frac{1}{2}k^2(k^2-2\gamma)Q^2_k + \frac{1}{2}\tilde{P}^2_k+\frac{1}{2}k^2(k^2-2\gamma)\tilde{Q}^2_k
\end{equation*}

as long as $2\gamma < k^2.$ The two particles, therefore, correspond to harmonic oscillators with excitation energies given by the well-known Bogoliubov-spectrum

\begin{equation} \label{47}
\Delta E_k=\sqrt{k^2(k^2-2\gamma)}.
\end{equation}

The summation over modes $k \in \mathbb{N}$ with $k^2>2\gamma$ is due to the high excitation energies \eqref{47} suppressed. We can hence approximate the summation by ignoring those excitations, i.e. we set $A_k=\ldots=D_k=0$ or $a_k=a_{-k}=0$ in \eqref{38}. This is precisely the condition on the possifold flow determining sensible possifolds we asked for below equation \eqref{42}. What does it imply?

For $\gamma<\frac{1}{2},$ the condition implies that all modes $k \in \mathbb{N}$ can be ignored. The possifold flow leads to a possifold collapsed to the point $\{ a_k=\sqrt{N} \delta_{k,0} | k \in \mathbb{Z} \}.$ It corresponds to the homogeneous condensate solution \eqref{41}. In other words, this means the ground state for $\gamma<\frac{1}{2}$ is well-described by the classical field \eqref{41}. 

The situation changes as we further increase $\gamma$ towards the critical point $\gamma=\frac{1}{2}.$ The summation over modes $a_{\pm 1}$ can no longer be ignored in the computation of observables. \eqref{41} ceases to be a good approximation of the ground-state. 

However, taking into account the lowest-lying modes $a_{\pm 1}$ leads to a sensible possifold. Concretely, it is given by the embedding

\begin{equation} \label{48}
S:=\{ a_k \in \mathbb{C} | k=\pm1, \sum_{k=\pm1} |a_k|^2 \le N,a_0=\sqrt{N-\sum_{k=\pm1} |a_k|^2}, a_{|k| \ge 2}=0 \}
\end{equation}    

in \eqref{38}.

In the standard Bogoliubov-approximation the Bogoliubov-transformation is encoded in the coordinates \eqref{46}. At the critical point $\gamma=\frac{1}{2},$ they break down. However, the possifold \eqref{48} gives rise to a regular theory even at the critical point. We describe it by returning to the regular coordinates \eqref{44} which we rephrase as

\begin{equation} \label{49}
\begin{split}
p_1 &= D_1 \\
q_1 &= A_1 \\
p_2 &= C_1 \\
q_2 &= B_1.
\end{split}
\end{equation}

The possifold flow then implies on $S$ 

\begin{equation} \label{50}
\Theta'=p_1 \delta q_1+p_2 \delta q_2
\end{equation}

and for the Hamiltonian

\begin{equation} \label{51}
H^{\prime}_{quadratic}=-\frac{\gamma N}{2}+\frac{1}{2}p_1^2+(\frac{1}{2}-\gamma)q_1^2+(\frac{1}{2}-\gamma)p_2^2+\frac{1}{2}q_2^2.
\end{equation}

Taking into account the higher order terms in \eqref{43}, we obtain

\begin{equation} \label{52}
\begin{split}
H^{\prime}_{higher}=&\frac{g}{16\pi}((q_1+q_2)^2+(p_1+p_2)^2+(q_1-q_2)^2+(p_1-p_2)^2) \cdot \\
&((q_1+q_2)(q_1-q_2)-(p_1+p_2)(p_1-p_2)) \\
&+\frac{g}{32\pi}(((q_1+q_2)^2+(p_1+p_2)^2)^2+((q_1-q_2)^2+(p_1-p_2)^2)^2 \\
&+((q_1+q_2)^2+(p_1+p_2)^2)((q_1-q_2)^2+(p_1-p_2)^2)).
\end{split}
\end{equation}

The possifold $S$ gives this way rise to a dual theory $(S,\Theta',H')$ describing the lowest-lying modes $a_{\pm 1}$ in \eqref{38}. As described, it is a valid approximation of the original theory up to the point when the higher modes $a_{\pm 2}$ have to be included. This regime of applicability includes especially the critical point $\gamma=\frac{1}{2}.$

Within our presented example, we have demonstrated the use of the formalism of the previous chapters. Before concluding this section, we briefly draw some conclusions of the dual theory obtained. 

Let $\gamma < \frac{1}{2}.$ The minimum of \eqref{51} is $q_{1,2}=p_{1,2}=0$ corresponding to the homogeneous condensate \eqref{41}. As already explained, this classical solution provides a good approximation to the ground state.

Let $\gamma > \frac{1}{2}.$ In this case, \eqref{51} changes its definiteness and the ground state is obtained by populating higher modes $a_{\pm k}$ in \eqref{37}. Instead of the homogeneous condensate \eqref{41}, the ground state is described by the bright soliton solution (see \cite{Kanamoto_2003} for further details).

The critical transition point $\gamma=\frac{1}{2}$ exhibits some interesting features as we already experienced. Here, \eqref{51} is degenerate in the directions $X_1=\partial_{q_1}$ and $X_2=\partial_{p_2}.$ From \eqref{23}, these vectorfields present symplectic symmetries. Because of $\mathcal{L}_{X_1} \Omega'=\mathcal{L}_{X_2} \Omega'=0$ and $\mathcal{L}_{X_1} H^{\prime}_{quadratic}=\mathcal{L}_{X_2} H^{\prime}_{quadratic}=0,$ the summation over paths in $S$ possesses an enhanced symmetry (up to $O(g)$-effects \eqref{52}). At the same time, this summation cannot be ignored at this critical point. We will see this behavior to occur in various other places.

The mentioned symmetry is broken explicitly by finite $g$ or by departure from the critical point $\gamma=\frac{1}{2}.$ It is exact in the Bogoliubov-limit

\begin{equation} \label{53}
g \to 0, N \to \infty, \gamma=\frac{gN}{2\pi}=\text{fixed}.
\end{equation}

While the standard Bogoliubov-approximation breaks down at the critical point $\gamma=\frac{1}{2},$ our dual theory stays regular and allows in principle a systematic treatment of finite $g$ effects beyond the limit \eqref{53}.

At the critical point, one might expect the dual theory obtained to be conformally invariant. Since it is in the limit \eqref{53} given by two non-relativistic free particles, it trivially has a $\mathfrak{sl}(2,\mathbb{R})$-invariance. However, the associated one dimensional conformal transformations are time-dependent. That is, in the summation over paths in $S,$ they are foliation-dependent. Symmetries and Ward-identities in the context of the summation \eqref{4} certainly deserve a treatment. We will here, however, not go further into this topic.

\subsection{Soliton-Possifold Correspondence}
\label{Kapitel 4.2}

In the last section, we have seen the possifold notion in an example. Keeping this concrete realization in mind, we want to go back to this general notion and clarify the physical meaning of a possifold in a given theory.

Consider the situation of definition \ref{Definition 3.1} for a given theory. The elements $\phi \in \Gamma=S \times T$ are of the form $\phi=(s,t)$ for $s \in S$ and $t \in T.$ According to \eqref{10}, $\phi$ corresponds to a solution of the equations of motion. 

In the example of the previous section this could have been the homogeneous condensate or the bright soliton solution. More generally, a solution of the equations of motion associated with the effective action which spontaneously breaks Poincar\'{e}-invariance is in \cite{Peskin:1995ev} defined to be a soliton. We relax the use of this definition a little and will denote as soliton simply a solution of the equations of motion. This does not fit quite the general notion. However, the ``solitons'' we will consider in practice will be mostly stationary, localized solutions like black holes or branes. In this way, we justify this little abuse of language in order to have a common word describing the situation. 

Consequently, $\phi=(s,t) \in \Gamma$ corresponds to a soliton. And again due to \eqref{10}, a shift in $s$ or in $t$ corresponds to an excitation of this soliton. For fixed $t \in T,$ the set $S \times \{t\} \subseteq \Gamma$ is an embedding of the possifold $S$ in $\Gamma.$ It describes an ensemble of excitations of the soliton $\phi.$ As far as observables concerning only those excitations are considered, they are by definition \ref{Definition 3.1} described via the theory associated to the possifold $S.$ 

We learn what the physical meaning of a possifold is. It can be thought as describing an ensemble of soliton excitations in a given theory. 

It is conversely natural to group the excitations of a given soliton of a theory into possifolds. We can think here of black holes, baryons, $D$- or $M$-branes as examples. The excitations of a soliton often possess a natural grouping. For instance, the fluctuations of a $D$-brane we would like to bundle in the manifold $S.$ The hope is that these excitations do have a description in terms of a theory, i.e. that $S$ is a possifold. Excitations like those of the gravitational field far away from the branes' throat are gathered in the manifold $T.$

We have seen this explicitly realized in the example of the last section. There, the possifold $S$ corresponded to excitations of the lowest modes of the soliton given by the homogeneous condensate. The manifold $T$ would correspond to excitations of the higher modes. 

Quite often we will identify a possifold $S$ with its embeddings describing an ensemble of suited soliton excitations. We will denote this identification as the \emph{soliton-possifold correspondence}. 

In the example of the last section, we derived the dual theory $(S,\Theta',H')$ associated to the possifold $S$ by carefully keeping track of the summation over weight factors. If a general possifold $(S,\Theta',H')$ in a given theory possesses a global canonical chart, we can write its partition function as in \eqref{8}

\begin{equation} \label{54}
Z_{\text{poss}}=\int \mathcal{D}q(t) e^{\frac{i}{\hbar} S_{\text{poss}}[q(t)]}
\end{equation}

for some degrees of freedom $q(t)$ and a possifold action $S_{\text{poss}}[q(t)].$ 

It would be nice if for some given soliton in a theory, \eqref{54} could be determined on general grounds. Precisely this happens in a well-known example - the most famous case of $D_3$-branes. 

\subsection{Sensitivity}
\label{Kapitel 4.3}

We end with an outlook on an example that appears to be interesting. A very obvious example of a soliton would be black holes in gravity. We translate the well-known properties of black holes in our developed framework. As we will see, it suggests the gravitational path integral to be rather intricate.

We want to consider four-dimensional Einstein-gravity. So far, we did not mention how to specify general relativity in terms of $(\Gamma,\Theta,H)$ as in chapter \ref{Kapitel 2}. We will here also not do so. The general connection with an action was extensively discussed in chapter \ref{Kapitel 2.2}. We will here only use those elements relevant for the argumentation to follow.

According to \eqref{10}, the space $\Gamma$ consists of solutions to the field equations. The potential $\Theta$ induces via \eqref{3} a volume form $\operatorname{Vol}$ on $\Gamma.$ \\
\\
In general, the volume form $\operatorname{Vol}$ on $\Gamma$ induces a notion of volume of subsets. For example, it associates a length to a path in $\Gamma.$ This implies the notion of \emph{distance} among states in $\Gamma.$ The physical meaning of these notions becomes apparent from the form of $\operatorname{Vol}$ expressed in local canonical coordinates as in \eqref{3}. Roughly, one can think of the volume as measuring the number of quantum mechanically resolvable states sectioned by a respective subset in $\Gamma.$\\
\\
In our present example, this is all the information about $\Theta$ we will use. The charge associated to the Hamiltonian is the mass of the spacetime. 

Gravity contains black holes and in light of the last section there is a natural possifold associated to them. It is the one capturing their microstates in accordance with their entropy. One might object whether such a possifold exists. It is surely required by consistency. We will take it simply as an assumption here. We could refer to string theory which associated particular black holes explicit possifolds. Therefore, we do not see any reason or principle to object the assumption made. 

We now add as matter content a number of $N$ scalar fields

\begin{equation*}
\Delta \mathcal{L} = \frac{1}{2} \sum_{i=1}^N \partial_\mu \phi_i \partial^\mu \phi_i
\end{equation*}

which we couple minimally to gravity. To make the later point clear, we imagine taking $N$ to be very large and finally taking $N \to \infty.$ We now construct randomly solutions to the field equations. We restrict them to be asymptotically flat and localized within a characteristic length $L.$ Additionally, we fix the charges associated to asymptotic Poincar\'{e}-invariance to vanish except the mass which we take to be $M.$ As mentioned, this equivalently produces an appropriate ensamble of states in $\Gamma.$ In statistical physics such an ensemble is known as a Liouville fluid. 

We proceed our thought experiment by monitoring the constructed fluid in $\Gamma$ for a given length scale $L.$ We continue by slowly decreasing $L \to R \sim GM$ towards the gravitational radius of $M.$ What is the behavior of the fluid observed in $\Gamma$?

For large $L \gg R,$ the distribution of the fluid is random. However, as $L$ approaches the gravitational radius $R,$ the fluid is by the black hole no-hair theorems forced to describe solutions diffeomorphic to the Schwarzschild black hole. For $L=R,$ the fluid is hence confined in a region describing the Schwarzschild black hole microstates with a volume of the order $\sim e^S$ with $S \sim \frac{R^2}{L^2_P}$ being the black hole entropy. Consequently, the enormous amount of solutions for large $L$ (remember that $N$ can be made arbitrarily large) has to flow into a finite volume in $\Gamma$ as $L$ is decreased towards $R.$ But this can happen in a smooth way only if the states are \emph{compressed} as $L$ is decreased. This is not a contradiction to Liouville's theorem as the process of changing $L$ does not happen in time and hence does not follow the volume preserving time evolution. In summary, we learn that the states of the fluid \emph{converge} towards one another as the length scale is decreased. The end point of this convergence is achieved when the decrease of $L$ reaches black hole formation and the fluid is confined inside an appropriate black hole possifold.

The argument suggests this behavior to be a glimpse of a general gravitational \emph{convergence theorem}. 

Indeed, the following altered thought experiment supports this anticipation. Consider again a Liouville fluid of states in $\Gamma.$ We require the corresponding solutions to the field equations be asymptotically flat and confined within a characteristic length scale $R.$ Additionally, we insist the charges associated to asymptotic Poincar\'{e}-invariance to not differ much from one another. Despite this innocent restrictions the fluid can be chosen arbitrarily and we denote its volume by $V.$ For each state in the fluid we can construct a new solution by adding an appropriate infalling mass shell at spatial infinity. The purpose of this shell is to collapse with the original field configuration to form a black hole of size $R.$ The solutions with added mass shells at spatial infinity lead to a new fluid in $\Gamma$ with volume $\sim V.$ The volume does not change significantly because the added mass shells of the various states do not differ much from one another. Following the time evolution of the new fluid, it flows by construction into the possifold of a black hole of size $R.$ Since this time the fluid follows the flow generated by $H,$ its volume stays constantly at $\sim V$ according to Liouville's theorem. And since the fluid enters completely the black hole's possifold, it is confined by a volume of order $\sim e^S$ with $S \sim \frac{R^2}{L^2_P}$ being the black hole entropy. Consequently, we obtain the bound

\begin{equation} \label{55}
\ln (V) \lesssim \frac{R^2}{L_P^2}.
\end{equation}

What is the statement of \eqref{55}? It looks like the Bekenstein bound but it says more. It proposes ``how'' gravity realizes the Bekenstein bound and the proposal is by convergence. Independently of how many scalar fields are used to create the solutions forming the Liouville fluid of volume $V,$ \eqref{55} puts a bound on this volume. This suggests that the states converge that much towards one another with respect to the measure $\operatorname{Vol}$ in order for \eqref{55} to be satisfied.

This suggestion seems to be in accordance with some conjectured properties quantum gravity should satisfy. A review on these properties is in \cite{Palti:2019pca} containing original references. Convergence provides a rationale for the distance conjecture. The bound on species appears as a special case of \eqref{55}. To make this clear, note that the statement of being able to ``resolve $N_s$ species at the scale $R$'' means that for each species, there is an associated independent field excitation. In the situation described before \eqref{55}, it therefore requires the existence of an appropriate Liouville fluid of volume $V \sim e^{N_s}.$ \eqref{55} is then the bound on species

\begin{equation} \label{56}
N_s \lesssim \frac{R^2}{L_P^2}.
\end{equation}

However, the presented arguments in this section have been heuristic and do not provide a proof of the convergence theorem. In fact, we even did not provide a rigorous statement of the convergence theorem within the structure $(\Gamma,\Theta,H).$ To do so, is beyond the scope of this small outlook section. We postpone a further more detailed analysis to a different place. 

Nevertheless, the suggested exponential suppression of the volume of shorter length scales may provide a remarkable possibility. Measurable quantities may have a weaker sensitivity on short length scales than one might expect from a usual perturbative field theory expansion. 

\section*{Acknowledgements}
We thank Alexander Gußmann for many discussions on this and other topics in physics.

\appendix

\section{Derivation of \eqref{2}}
\label{Appendix A}

In the quantum theory of chapter \ref{Kapitel 2}, we are typically interested in the probability amplitude

\begin{equation} \label{A1}
\left\langle q_b \right| e^{-\frac{i}{\hbar}HT} \left| q_a \right\rangle
\end{equation}

for the system to evolve from the state $\left| q_a \right\rangle$ to $\left| q_b \right\rangle$ within a time $T.$ To get an expression for \eqref{A1}, we split the time interval $T=n \varepsilon$ into a large number $n$ of small intervals of size $\varepsilon$ and decompose the time evolution 

\begin{equation} \label{A2}
e^{-\frac{i}{\hbar}HT} = \underbrace{e^{-\frac{i}{\hbar}H\varepsilon} \cdots e^{-\frac{i}{\hbar}H\varepsilon}}_{n\text{-times}}.
\end{equation}

For $k=1,\ldots,n-1$ we sandwich the completeness relation

\begin{equation} \label{A3}
1 = \left( \prod_{i=1}^N {\int {dq_k^i}} \right) \left| q_k \right\rangle \left\langle q_k \right|
\end{equation}

between the exponentials in \eqref{A2}. Defining $q_0 := q_a$ and $q_n := q_b$ we are left to evaluate expressions of the form

\begin{equation} \label{A4}
\left\langle q_{k+1} \right| e^{-\frac{i}{\hbar}H\varepsilon} \left| q_k \right\rangle \xrightarrow[\varepsilon \to 0]{}
\left\langle q_{k+1} \right|  1 -\frac{i}{\hbar}H\varepsilon + \ldots \left|  q_k \right\rangle  
\end{equation}  

for $k=0,\ldots,n-1.$  Suppose, between the bra and ket in \eqref{A4} we would have an operator $f(q)$ that is a function of the position operators only.  We then have

\begin{equation} \label{A5}
\begin{split}
\left\langle q_{k+1} \right| f(q) \left| q_k \right\rangle
&= f \left( q_k \right) \prod_{i=1}^N {\delta \left( q_k^i-q_{k+1}^i \right)} \\
&= f\left( \frac{q_{k+1}+q_k}{2} \right) \left( \prod_{i=1}^N {\int \frac{dp_k^i}{2\pi \hbar}} \right) e^{\frac{i}{\hbar} \sum_{i=1}^N {p_k^i \left( q_{k+1}^i-q_k^i\right)}}.
\end{split}
\end{equation}

In case $f$ is only momentum dependent, we can write

\begin{equation} \label{A6}
\begin{split}
\left\langle q_{k+1} \right| f(p) \left| q_k \right\rangle
&= \left( \prod_{i=1}^N {\int \frac{dp_k^i}{2\pi \hbar}} \right) f ( p_k ) \langle q_{k+1} |  p_k \rangle \langle p_k | q_k \rangle \\
&=  \left( \prod_{i=1}^N {\int \frac{dp_k^i}{2\pi \hbar}} \right) f ( p_k ) e^{\frac{i}{\hbar} \sum_{i=1}^N {p_k^i \left( q_{k+1}^i-q_k^i\right)}}.
\end{split}
\end{equation}

The above formulae are summarized in the expression

\begin{equation} \label{A7}
\left\langle q_{k+1} \right| f(q,p) \left| q_k \right\rangle
= \left( \prod_{i=1}^N {\int \frac{dp_k^i}{2\pi \hbar}} \right) f  \left( \frac{q_{k+1}+q_k}{2},p_k \right) e^{\frac{i}{\hbar} \sum_{i=1}^N {p_k^i \left( q_{k+1}^i-q_k^i\right)}}
\end{equation}

which is, however, not valid in general. What are the conditions for its validity?

Let $O: \Gamma \longrightarrow \mathbb{R}$ be a classical observable. We define a proper symmetrization map $S$ that associates to $O$ an Weyl-ordered quantum mechanical operator $O^S$ in the Schrödinger-picture. We define $S$ by insisting monoms to be mapped as

\begin{equation} \label{A8}
q^\alpha f(p) \xrightarrow[]{S} \sum_{k=0}^\infty {\binom{\alpha}{k} \left( \frac{q}{2} \right)^k f(p) \left( \frac{q}{2} \right)^{\alpha - k}}
\end{equation}

for $\alpha \in \mathbb{R}$ and an arbitrary function $f$ of the momentum. Note that the left-hand side in \eqref{A8} denotes a classical observable $\Gamma \longrightarrow \mathbb{R}$ while the right-hand side is a quantum mechanical operator. Obviously, this operator is Weyl-ordered and by the argumentation as in \eqref{A5} and \eqref{A6} together with the generalized binomial formula, the validity of \eqref{A7} is implied for operators as in \eqref{A8}. $S$ is then uniquely determined by requiring linearity and continuity. We conclude \eqref{A7} to hold in the form

\begin{equation} \label{A9}
\left\langle q_{k+1} \right| O^S (q,p) \left| q_k \right\rangle
= \left( \prod_{i=1}^N {\int \frac{dp_k^i}{2\pi \hbar}} \right) O  \left( \frac{q_{k+1}+q_k}{2},p_k \right) e^{\frac{i}{\hbar} \sum_{i=1}^N {p_k^i \left( q_{k+1}^i-q_k^i\right)}}
\end{equation}

for an arbitrary observable $O=O(q,p).$ As is customary, we will call $O^S$ the quantum mechanical operator associated to the classical observable $O.$ Moreover, we will in general omit the superscript $S.$ By using the canonical commutators, each operator can be written in Weyl-ordered form. Hence, the map $S$ is surjective. Without loss of generality, we therefore assume the Hamilton-operator in \eqref{A1} to be Weyl-ordered and being the image of the Hamilton-function $H=H(q,p)$ under $S.$ 

With these observations, we can evaluate \eqref{A4} as

\begin{equation} \label{A10}
\left\langle q_{k+1} \right| e^{-\frac{i}{\hbar}H\varepsilon} \left| q_k \right\rangle
= \left( \prod_{i=1}^N {\int \frac{dp_k^i}{2\pi \hbar}} \right) e^{-\frac{i}{\hbar} \varepsilon  H \left( \frac{q_{k+1}+q_k}{2},p_k \right) } e^{\frac{i}{\hbar} \sum_{i=1}^N {p_k^i \left( q_{k+1}^i-q_k^i\right)}}.
\end{equation}     

Collecting everything together, this implies finally

\begin{equation} \label{A11}
\begin{split}
&\left\langle q_b \right| e^{-\frac{i}{\hbar}HT} \left| q_a \right\rangle\\
&= \left( \prod_{k=1}^{n-1} \prod_{i=1}^N \int dq_k^i \left\langle q_{k+1} \right| e^{-\frac{i}{\hbar}H\varepsilon} \left| q_k \right\rangle \right) \left\langle q_{1} \right| e^{-\frac{i}{\hbar}H\varepsilon} \left| q_0 \right\rangle \\
&= \prod_k \prod_{i=1}^N \int dq_k^i \int \frac{dp_k^i}{2\pi \hbar} e^{\frac{i}{\hbar} \sum_{k=0}^{n-1} \left( \sum_{i=1}^N p_k^i \left( q_{k+1}^i-q_k^i\right) - \varepsilon  H \left( \frac{q_{k+1}+q_k}{2},p_k \right) \right)}.
\end{split}
\end{equation}

\eqref{A11} is the standard path integral expression for the solution of the Schrödinger equation satisfied by \eqref{A1}. We rederived it here in order to make clear the notation used in what follows. We can think of $k$ as labeling the time-slice with time $t=k \varepsilon.$ At the slice $t$ in our time-foliation we have the measure $\mathcal{D} q(t) = \prod_{i=1}^N dq^i.$ The integration over $\mathcal{D}q(t)$ occurs in every slice where the position is not fixed. Furthermore, we have the measure $\mathcal{D}p(t) = \prod_{i=1}^N \frac{dp^i}{2\pi \hbar}.$ Integration over $\mathcal{D}p(t)$ occurs for each transition from a time-slice $t$ to the next one. 

With this notation we can write \eqref{A11} more compactly as

\begin{equation} \label{A12}
\left\langle q_b \right| e^{-\frac{i}{\hbar}HT} \left| q_a \right\rangle
= \int_{q(t=0)=q_a}^{q(t=T)=q_b} \mathcal{D}q(t) \mathcal{D}p(t) e^{\frac{i}{\hbar} \int_0^T dt \left(\sum_i p^i \dot{q}^i -H(q,p) \right)}.
\end{equation}

The limits on the integrals mean that on the particular time-slices the position is fixed implying that there is no need of integration. \\
\\
Having a functional integral expression for the probability amplitude enables us to find an expression for the correlation functions. Let $O_i : \Gamma \longrightarrow \mathbb{R}$ for $i=1,2$ be two observables. Let further $t_i \in (-T,T).$ Without loss of generality, we assume the ordering $t_2 > t_1.$ We then consider the expression

\begin{equation} \label{A13}
\begin{split}
\int_{q(-T)=q_a}^{q(T)=q_b} \mathcal{D}q(t) \mathcal{D}p(t) &O_1(q(t_1),p(t_1)) O_2(q(t_2),p(t_2)) \cdot \\
&e^{\frac{i}{\hbar} \int_{-T}^T dt \left(\sum_i p^i(t) \dot{q}^i(t) -H(q(t),p(t)) \right)}.
\end{split}
\end{equation} 

In the following, we omit writing the limits for $q(\pm T)$ and use the abbreviation $O_i(q(t_i),p(t_i))=O_i(t_i)=O_i$ to simplify notation. Let be $\delta > 0$ small enough such that we can split \eqref{A13} as

\begin{equation} \label{A14}
\begin{split}
&\int_{t>t_2+\delta} \mathcal{D}q(t) \mathcal{D}p(t) e^{\frac{i}{\hbar} \int_{t_2+\delta}^T dt \left(\sum_i p^i \dot{q}^i -H \right)} \times \\
&\int_{t_2+\delta>t>t_2-\delta} \mathcal{D}q(t) \mathcal{D}p(t) O_2(t_2) e^{\frac{i}{\hbar} \int_{t_2-\delta}^{t_2+\delta} dt \left(\sum_i p^i \dot{q}^i -H \right)} \times \\
&\int_{t_2-\delta>t>t_1+\delta} \mathcal{D}q(t) \mathcal{D}p(t) e^{\frac{i}{\hbar} \int_{t_1+\delta}^{t_2-\delta} dt \left(\sum_i p^i \dot{q}^i -H \right)} \times \\
&\int_{t_1+\delta>t>t_1-\delta} \mathcal{D}q(t) \mathcal{D}p(t) O_1(t_1) e^{\frac{i}{\hbar} \int_{t_1-\delta}^{t_1+\delta} dt \left(\sum_i p^i \dot{q}^i -H \right)} \times \\
&\int_{t<t_1-\delta} \mathcal{D}q(t) \mathcal{D}p(t) e^{\frac{i}{\hbar} \int_{-T}^{t_1-\delta} dt \left(\sum_i p^i \dot{q}^i -H \right)}.
\end{split}
\end{equation}  

For the parts without insertions of $O_i$ in the integrands, we can apply \eqref{A12}. The others are rewritten using \eqref{A9} and the smoothness of the path $(q(t),p(t))$ in time. This smoothness is a consequence of the construction of the functional integral. We obtain for \eqref{A14}

\begin{equation} \label{A15}
\begin{split}
&\int \mathcal{D}q(t_2+\delta) \left\langle q_b \right| e^{-\frac{i}{\hbar}H(T-(t_2+\delta))} \left| q(t_2+\delta) \right\rangle \times \\
&\int \mathcal{D}q(t_2-\delta) \left\langle q(t_2+\delta) \right| \left(O_2 e^{-\frac{i}{\hbar}H \cdot 2\delta} \right)^S \left| q(t_2-\delta) \right\rangle \times \\
&\int \mathcal{D}q(t_1+\delta) \left\langle q(t_2-\delta) \right| e^{-\frac{i}{\hbar}H(t_2-t_1-2\delta)} \left| q(t_1+\delta) \right\rangle \times \\
&\int \mathcal{D}q(t_1-\delta) \left\langle q(t_1+\delta) \right| \left( O_1 e^{-\frac{i}{\hbar}H \cdot 2\delta} \right)^S \left| q(t_1-\delta) \right\rangle \times \\
&\left\langle q(t_1-\delta) \right| e^{-\frac{i}{\hbar}H(t_1-\delta+T)} \left| q_a \right\rangle.
\end{split}
\end{equation}

Using the completeness relation and taking the limit $\delta \to 0$ the expression becomes

\begin{equation} \label{A16}
\begin{split}
&\left\langle q_b \right| e^{-\frac{i}{\hbar}HT} e^{-\frac{i}{\hbar}H(-t_2)} O_2^S e^{-\frac{i}{\hbar}Ht_2}
e^{-\frac{i}{\hbar}H(-t_1)} O_1^S e^{-\frac{i}{\hbar}Ht_1} e^{-\frac{i}{\hbar}HT} \left| q_a \right\rangle \\
&= \left\langle q_b \right| e^{-\frac{i}{\hbar}HT} T \left( O_2^H(t_2) O_1^H(t_1) \right) e^{-\frac{i}{\hbar}HT} \left| q_a \right\rangle.
\end{split}
\end{equation}

In the last equation, we used the relation $O^H(t) = e^{\frac{i}{\hbar}Ht} O^S e^{-\frac{i}{\hbar}Ht}$ among Heisenberg-picture and Schrödinger-picture operators. The correct time-order is an immediate consequence of the derivation. 

The remaining part proceeds as in the derivation of the standard case. We can expand $H=E_0 |\Omega \rangle \langle \Omega| + \sum_n E_n |n \rangle \langle n|$ in a complete orthonormal basis with $| \Omega \rangle$ being the lowest eigenvalue state. As already noted, the path $(q(t),p(t))$ in the functional integration is smooth and can hence be analytically continued in $t.$ This allows taking the limit $T \to \infty (1-i\varepsilon)$ for a small $\varepsilon > 0.$ The least suppressed part in \eqref{A16} then becomes

\begin{equation} \label{A17}
\langle q_b | \Omega \rangle e^{-2\frac{i}{\hbar} E_0 \infty(1-i\varepsilon)} \langle \Omega | q_a \rangle 
\langle \Omega | TO_2^H (t_2) O_1^H(t_1) |\Omega \rangle .
\end{equation}

We can now add an integration over $q_a$ and $q_b$ to include the position integration at $\pm T$ in \eqref{A13}. Had we performed the same derivation with no operator insertions in the integrand, we would have obtained \eqref{A17} without the last factor. Dividing both results thus yields precisely \eqref{2} for the case of two operator insertions. The general case is derived analogously.

\end{document}